\documentclass[12pt]{article}
\title{
Quantum diffusion in the Kronig-Penney model
}
\author{Masahiro Kaminaga 
\thanks{
Department of Electrical Engineering and Information Technology, 
Tohoku Gakuin University, Tagajo, 985-8537, JAPAN. 
Tel.: +81-22-368-7059 \ 
E-mail:kaminaga@mail.tohoku-gakuin.ac.jp}           
\and
Takuya Mine
\thanks{
Faculty of Arts and Sciences,
Kyoto Institute of Technology,
Matsugasaki, Kyoto, 606-8585, JAPAN.
Tel: +81-75-724-7834 \ 
E-mail: mine@kit.ac.jp
}
}

\date{Received: date / Accepted: date}
\usepackage{amsthm,amssymb,amsmath,graphicx, mathrsfs}

\newtheorem{theorem}{\indent Theorem}
\newtheorem{proposition}[theorem]{\indent Proposition}
\newtheorem{lemma}[theorem]{\indent Lemma}

\def\Im{\mathop{\rm Im}\nolimits}

\def\tr{\mathop{\rm tr}\nolimits}
\def\det{\mathop{\rm det}\nolimits}

\begin{document}
\maketitle

\begin{abstract}
In this paper we consider the 
1D Schr\"odinger operator $H$ with periodic point interactions.
We show an $L^1-L^\infty$ bound for the 
time evolution operator $e^{-itH}$ restricted to each energy band
with decay order $O(t^{-1/3})$ as $t\to \infty$,
which comes from some kind of resonant state.
The order $O(t^{-1/3})$ is optimal for our model.
We also give an asymptotic bound for the coefficient
in the high energy limit.
For the proof,
we give an asymptotic analysis for the band functions
and the Bloch waves in the high energy limit.
Especially we give the asymptotics for the inflection points
in the graphs of band functions,
which is crucial for the asymptotics of the coefficient
in our estimate.  
\end{abstract}

\section{Introduction}
\label{intro}
The one-electron models of solids are based on the study of 
Schr\"odinger operator with periodic potential.
There are a lot of studies on the periodic potential, 
in particular, for periodic point interactions, 
we can show the spectral set explicitly 
(Albeverio et.\ al.\ \cite{Alb2nd} is the best guide to this field for readers).
Most fundamental case is the one-dimensional Schr\"odinger operator 
with periodic point interactions,
called the \textit{Kronig--Penney model} 
(see Kronig--Penney \cite{KrPe}),
given by
\begin{equation}
\label{kronig-penney0}
H = - \frac{d^2}{dx^2} + V\sum_{j=-\infty}^{\infty}\delta(x-j) 
\quad \mbox{on }L^2({\mathbf R}),
\end{equation}
where $V$ is a non-zero real constant, and $\delta(\cdot - j)$ is 
the Dirac delta measure at $j\in{\mathbf Z}$. 
The positive sign of $V$ corresponds the repulsive interaction, 
while the negative one corresponds the attractive one.
More precisely, $H$ is the negative Laplacian 
with boundary conditions on integer points:
\begin{equation}
\label{kronig-penney1}
H = -\frac{d^2}{dx^2}\quad \mbox{on }{\cal D}, 
\end{equation}
where 
\begin{equation}
\label{kronig-penney2}
{\cal D}=\{u\in H^1({\mathbf R})\cap H^2({\mathbf R}\setminus{\mathbf Z})
:u'(j+)-u'(j-)=Vu(j),\ j\in{\mathbf Z}\}.
\end{equation}
Here 
$u(j\pm)=\lim_{\epsilon \to +0}u(j\pm \epsilon)$ and
$H^p(\Omega)$ is the usual Sobolev space of order $p$ 
on the open set $\Omega$. 
From Sobolev's embedding theorem 
$H^1({\mathbf R})\hookrightarrow C_b^0({\mathbf R})$, 
every elements of ${\cal D}$ are continuous (classical sense) 
and uniformly bounded functions.
It is well-known that $H$ is self-adjoint \cite{{GeKir},{Alb2nd}} and 
is a model describing electrons on the quantum wire.
The spectrum of this model is explicitly given by 
\begin{equation}
\label{kronig-penney3}
\sigma(H) = \left\{E\in\mathbf{R}:-2\leq D(\sqrt{E})\leq 2 \right\},
\end{equation}
where 
$$
D(k) = 2\cos k + V\frac{\sin k}{k}.
$$
$D(k)$ is so-called {\it discriminant} 
and $D(\sqrt{E})$ can be regarded as an entire function 
with respect to $E\in \mathbf{C}$.
The spectrum $\sigma(H)$ of $H$ consists of infinitely 
many closed intervals (spectral bands) 
and is purely absolutely continuous. 

On the other hand, 
for the Schr\"odinger operator $H=-\Delta + V$ on $\mathbf{R}^d$
with decaying potential $V$,
the \textit{dispersive estimate}
for the Schr\"odinger time evolution operator
$e^{-itH}$ is stated as follows:
\begin{eqnarray}
\label{dispersive_estimate}
&& \|P_{ac}e^{-it H} u\|_{L^p(\mathbf{R}^d)}
\leq C |t|^{-d\left(\frac{1}{2}-\frac{1}{p}\right)} \|u\|_{L^q(\mathbf{R}^d)},
\quad
u\in L^2(\mathbf{R}^d)\cap L^q(\mathbf{R}^d)\\
&&(2<p\leq\infty,\ 1/p+1/q=1),\nonumber
\end{eqnarray}
where $P_{ac}$ denotes the spectral projection 
to the absolutely continuous subspace for $H$.
The dispersive estimate
is a quantitative representation of the
diffusion phenomena in quantum mechanics,
and is extensively studied recently,
because of its usefulness in the theory of 
the non-linear Schr\"odinger operator
(see e.g.\ 
Journ\'e--Soffer--Sogge \cite{JoSoSo},
Weder \cite{We},
Yajima \cite{Ya},
Mizutani \cite{Miz},
and references therein).
The estimate (\ref{dispersive_estimate})
is also obtained 
in the case of the one-dimensional point interaction.
Adami--Sacchetti \cite{AdSa} obtain (\ref{dispersive_estimate})
when $V$ is one point $\delta$ potential,
and so do Kova{\v r}{\'\i}k--Sacchetti \cite{KoSa}
when $V$ is the sum of $\delta$ potentials at two points.
The motivation of the present paper is
to obtain a similar estimate for our periodic model
(\ref{kronig-penney0}).
Though this problem is quite fundamental,
we could not find such kind of results
in the literature,
probably because the deduction of the result
requires a detailed analysis
of the band functions, as we shall see below.

Since the spectrum of 
the Schr\"odinger operator with periodic potential 
is absolutely continuous,
one may expect some dispersion type estimate holds
also in this case.
However, there seems to be few results 
about the dispersive type estimate for the time
evolution operator of the differential equation
with periodic coefficients.\footnote{%
Some authors study the pointwise asymptotics
for the integral kernel
of the time evolution operator in the large time limit;
see e.g.\ Korotyaev \cite{Ko} and references therein.
But we could not find the dispersive type estimate 
for the periodic Schr\"odinger operator itself
in the literature.
}
An example is the paper by
Cuccagna \cite{Cuc},
in which the Klein--Gordon equation
\begin{eqnarray*}
&& u_{tt} + H u + \mu u = 0,\quad H = -\frac{d^2}{dx^2}+P(x)
\quad \mbox{on } \mathbf{R},\\
&&u(0,x)=0,\quad u_t(0,x)=g(x)
\end{eqnarray*}
is considered, where $P(x)$ is a smooth real-valued 
periodic function with period $1$.
Cuccagna proves the solution $u(t,x)$ satisfies
\begin{equation}
\label{cuccagna}
 \|u(t,\cdot)\|_{L^\infty(\mathbf{R})}
\leq C_\mu \langle t\rangle^{-1/3}\|g(\cdot)\|_{W^{1,1}(\mathbf{R})},\quad
\langle t \rangle=\sqrt{1+t^2}
\end{equation}
for $\mu\in (0,\infty)\setminus D$,
where $D$ is some bounded discrete set.
The peculiar power $-1/3$ comes from the following reason.
The integral kernel for the time evolution operator
is written as the sum of oscillatory integrals
\begin{equation}
\label{intro0}
 K_{n,t}(x,y)= \frac{1}{2\pi}
\int_{-\pi}^\pi e^{-it (\lambda_n(\theta)-s\theta)}a_n(\theta,x',y')\, d\theta\quad
(n=1,2,\ldots),
\end{equation}
where 
$x=[x]+x'$, $y=[y]+y'$, $[x],[y]\in\mathbf{Z}$, $x',y'\in(0,1)$,
and $s=([x]-[y])/t$.
The function $\lambda_n(\theta)$
is called the \textit{band function} for the $n$-th band,
which is a real-analytic function of $\theta$ with period $2\pi$.
In the large time limit $t\to \infty$,
it is well-known that the main contribution of the oscillatory
integral (\ref{intro0}) comes from the part
nearby the \textit{stationary phase point} $\theta_s$
(the solution of $\lambda_n'(\theta)=s$),
and the \textit{stationary phase method} 
tells us the principal term in the asymptotic bound
is a constant multiple of 
$|\lambda_n''(\theta_s)|^{-1/2}t^{-1/2}$ 
(see e.g.\ Stein \cite[Chapter VIII]{St} or Lemma \ref{lemma_stein} below).
However, since $\lambda_n(\theta)$ is a periodic function,
there exists a point $\theta_0$ so that $\lambda_n''(\theta_0)=0$.
If the stationary phase point $\theta_s$ coincides with $\theta_0$,
then the previous bound no longer makes sense.
Instead, the stationary phase method concludes
the principal term is a constant multiple of
$|\lambda_n^{(3)}(\theta_0)|^{-1/3} t^{-1/3}$.
Since the integral kernel of our operator has the same form 
as (\ref{intro0})
(see (\ref{st01}) below),
we expect a result similar to (\ref{cuccagna}) also holds in our case.

Let us formulate our main result.
Let $H$ be the Hamiltonian for the Kronig-Penney model
given in (\ref{kronig-penney1}) and (\ref{kronig-penney2}).
As stated above,
the spectrum of $H$ has the band structure, that is,
\begin{displaymath}
 \sigma(H) = \bigcup_{n=1}^\infty I_n,
\end{displaymath}
where the $n$-th band $I_n$ is a closed interval of finite length
(for the precise definition,
see (\ref{bloch13}) below).
Our main result is as follows.
\begin{theorem}
\label{theorem_main}
Let $P_n$ be the spectral projection onto the 
$n$-th energy band $I_n$.
Then, for sufficiently large $n$,
there exist positive constants
$C_{1,n}$ and $C_{2,n}$ such that
\begin{equation}
\label{main1}
 \left\|P_n e^{-it H} u\right\|_{L^\infty}
\leq (C_{1,n}\langle t \rangle^{-1/2} + C_{2,n} \langle t \rangle^{-1/3}) \|u\|_{L^1}
\end{equation}
for any $u \in L^1(\mathbf{R})$ and any $t\in \mathbf{R}$,
where $\langle t \rangle=\sqrt{1+t^2}$.
The coefficients obey the bound
\begin{displaymath}
C_{1,n} = O(1),\quad
C_{2,n} = O(n^{-1/9})
\end{displaymath}
as $n\to\infty$.
\end{theorem}
\noindent
The power $-1/2$ in the first term of the coefficient in (\ref{main1}) 
is the same as in (\ref{dispersive_estimate}) with $d=1$ and $p=\infty$,
since it comes from the states corresponding to the energy
near the band center, which behaves like a free particle.
This fact can be understood from the graph of $\lambda=\lambda_n(\theta)$
(Figure \ref{figure_band_function}).%
\footnote{All the graphs are written by using Mathematica 9.0.}
The part of the graph corresponding to the band center 
is similar to the parabola $\lambda=\theta^2$ or its translation,
which is the band function for the free Hamiltonian $H_0=-d^2/dx^2$.
On the other hand,
the power $-1/3$ comes from part of the integral (\ref{intro0})
given by
\begin{equation}
\label{intro1}
 \tilde{K}_{n,t}(x,y)= \frac{1}{2\pi}
\int_{J_n} e^{-it (\lambda_n(\theta)-s\theta)}a_n(\theta,x',y')\, d\theta,
\end{equation}
where $J_n$ is some open set including 
two solutions $\theta_0$ to the equation $\lambda_n''(\theta)=0$.
Notice that $(\theta_0,\lambda_n(\theta_0))$ is an inflection 
point in the graph of $\lambda=\lambda_n(\theta)$;
see Figure \ref{figure_band_function}.
The estimates for the coefficients 
are obtained from the lower bounds for the derivatives 
of $\lambda_n(\theta)$.
Actually, we can choose $J_n$ so that
\begin{equation}
\label{intro2}
 \inf_{\theta\in [-\pi,\pi]\setminus J_n}|\lambda_n''(\theta)|\geq C,\quad
  \inf_{\theta\in J_n}|\lambda_n^{(3)}(\theta)|\geq C n^{1/3},
\end{equation}
where $C$ is a positive constant independent of $n$.
By (\ref{intro2}) and the estimates for the amplitude function
(Proposition \ref{proposition_amplitude1}),
we can prove Theorem \ref{theorem_main} by using
a lemma for estimating oscillatory integrals,
given in Stein's book
(see Stein \cite[page 334]{St} or Lemma \ref{lemma_stein} below).
We can also prove the power $-1/3$ is optimal,
by considering the case $s=\lambda'(\theta_0)$
(so, $\theta_0$ is a stationary phase point),
and applying the asymptotic expansion formula in the stationary phase method
(see e.g.\ Stein \cite[Page 334]{St}).
\begin{figure}[htbp]
 \begin{center}
  \includegraphics[width=8cm]{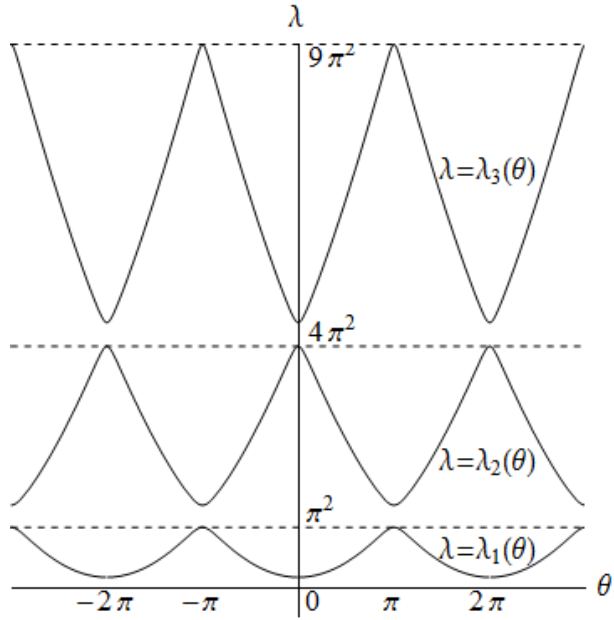}
 \end{center}
\caption{The graphs of the band functions 
$\lambda=\lambda_n(\theta)$ ($n=1,2,3$)
for $V>0$.
The range of $\lambda_n(\theta)$ is the $n$-th band $I_n$.
We find two inflection points of 
$\lambda=\lambda_n(\theta)$ near $(n\pi,(n\pi)^2)$,
for every $n$.
}
\label{figure_band_function}
\end{figure}

The physical implication of the result is as follows.
By definition, the parameter $s=([x]-[y])/t$ represents
the \textit{propagation velocity} of a quantum particle.
The wave packet with energy near $\lambda(\theta_0)$ has
the maximal group speed in the $n$-th band,
and the speed of the quantum diffusion is slowest in that band.
Thus such state has a bit longer life-span 
(in the sense of $L^\infty$-norm)
than the ordinary state has;
the state is in some sense a \textit{resonant state},
caused by the meeting of two stationary phase points $\theta_s$
as $s$ tends to $s_0=\lambda'(\theta_0)$.
It is well-known that the existence of resonant states
makes the decay of the solution 
with respect to $t$ slower 
(see Jensen--Kato \cite{JeKa} or Mizutani \cite{Miz}).

Since the estimate (\ref{main1}) is given bandwise,
it is natural to ask we can obtain 
the dispersive type estimate for 
the whole Schr\"odinger time evolution operator $e^{-itH}$,
like Cuccagna's result (\ref{cuccagna}).
However, it turns out to be difficult
in the present case, from the following reason.
A reasonable strategy to prove such estimate is as follows.
First, we divide the integral $K_{n,t}$ into two parts $\tilde{K}_{n,t}$
and the rest, where $\tilde{K}_{n,t}$ is given in (\ref{intro1})
with $J_n$ some open set including $\theta=0, \pm\pi$.
Next, we show the sum of $\tilde{K}_{n,t}$ converges and gives $O(t^{-1/3})$, 
and the sum of the rests also converges and gives $O(t^{-1/2})$.
However, for fixed $x$ and $y$, we find that our upper bound
for $\tilde{K}_{n,t}(x,y)$ is not better than $O(n^{-1}t^{-1/2})$,
and the sum of the upper bounds does not converge
(see the last part of Section 4).
One reason for this divergence is very strong singularity
of our potential, the sum of $\delta$-functions.
Because of this singularity,
the width of the band gap,
say $g_n$ ($n$ is the band number), 
does not decay at all in the high energy limit $n\to \infty$.%
\footnote{
Proposition \ref{proposition_puiseux_lambda}
implies $g_n \to 2|V|$ as $n\to \infty$.}
Then we cannot take the open set $J_n$ so small,%
\footnote{
We take $|J_n|=O(n^{-1/3})$ in the proof of Theorem \ref{theorem_main}.
}
and the sum of the lengths $|J_n|$ diverges;
if this sum converges,
we can use a simple bound $|\tilde{K}_{n,t}|\leq C|J_n|$
to control the sum.
Thus we do not succeed to obtain a bound for
the sum of $\tilde{K}_{n,t}(x,y)$ at present.

On the other hand, for the Schr\"odinger operator $H=-d^2/dx^2 + V$
on $\mathbf{R}^1$ with real-valued periodic potential $V$,
it is known that the decay rate of the width of the band gap $g_n$
reflects the smoothness of the potential $V$.
Hochstadt \cite{Ho} says $g_n = o(n^{-(m-1)})$ if $V$ is in $C^m$,
and Trubowitz \cite{Tr} says $g_n = O(e^{-c n})$ 
($c$ is some positive constant) if $V$ is real analytic.
So, if $V$ is sufficiently smooth,
it is expected that we can control the sum of $\tilde{K}_{n,t}$,
and obtain the dispersive type estimate for the whole operator $e^{-itH}$
(i.e.\ (\ref{main1}) without the projection $P_n$).
We hope to argue this matter elsewhere in the near future.

The paper is organized as follows.
In Section 2,
we review the Floquet--Bloch theory for our operator $H$
and give the explicit form of the integral kernel
of $e^{-itH}$.
In Section 3,
we give more concrete analysis for the band functions,
especially give some estimates for the derivatives.
In Section 4, we prove Theorem \ref{theorem_main},
and give some comment for the summability with respect to $n$
of the estimates (\ref{main1}).

\section{Floquet--Bloch theory}
In this section, we shall calculate the integral kernel
of the operator $e^{-itH}$ by using the Floquet-Bloch theory.
Most results in this section are already written in another literature
(e.g.\ Reed-Simon \cite[XIII.16]{ReSi4} and Albeverio et.\ al.\ \cite[III.2.3]{Alb2nd}),
but we shall give it here again for the completeness.

First we shall calculate the generalized eigenfunctions
for our model,
i.e., the solutions to the equations
\begin{eqnarray}
&& -\varphi''(x) = \lambda \varphi(x)\quad (x\in \mathbf{R}\setminus\mathbf{Z}),
\label{bloch01}\\
&&\varphi(j+)=\varphi(j-)\quad (j\in \mathbf{Z}),
\label{bloch02}\\
&&\varphi'(j+)-\varphi'(j-)=V \varphi(j)\quad (j\in \mathbf{Z}).
\label{bloch03}
\end{eqnarray}
The condition (\ref{bloch02}) comes from the requirement 
$\varphi\in H^1_{\rm loc}(\mathbf{R})$, and we use the abbreviation
$\varphi(j)=\varphi(j\pm)$ in (\ref{bloch03}).
\begin{proposition}
 \label{proposition_generalized_eigenfunction}
Let $\lambda\in \mathbf{C}$, $V\in \mathbf{R}$,
and take $k\in \mathbf{C}$ so that $\lambda=k^2$.
Then, the equations (\ref{bloch01})-(\ref{bloch03}) 
have a solution $\varphi(x)$ of the following form.
\begin{equation}
\label{bloch05}
 \varphi(x)= A_j \cos k(x-j) +  B_j k^{-1} \sin k(x-j)\quad
(j<x<j+1),
\end{equation}
where $A_j$ and $B_j$ are constants.
When $k=0$, we interpret $k^{-1}\sin k(x-j) = x-j$.
The coefficients $A_j$ and $B_j$ satisfy the following
recurrence relation.
\begin{equation}
\label{bloch06}
\begin{pmatrix}
 A_j\cr B_j
\end{pmatrix}
= T(k)
\begin{pmatrix}
 A_{j-1}\cr B_{j-1}
\end{pmatrix}
\quad (j\in \mathbf{Z}),
\end{equation}
\begin{displaymath}
 T(k)=
\begin{pmatrix}
\cos k & k^{-1}\sin k\cr
V \cos k - k \sin k & V k^{-1}\sin k+\cos k
\end{pmatrix}.
\end{displaymath}
The matrix $T(k)$ satisfies
$\det T(k)=1$ and the discriminant $D(k)=\tr T(k)$ is
\begin{eqnarray}
\label{bloch08} 
D(k)=2 \cos k + V k^{-1}\sin k.
\end{eqnarray}
\end{proposition}
\noindent
The proof is a simple calculation, so we shall omit it.
Notice that $D(k)=D(\sqrt{\lambda})$ is an entire function
with respect to $\lambda$, since $D(k)$ is an even function.

Next  we shall calculate the \textit{Bloch waves},
the solution to (\ref{bloch01})-(\ref{bloch03})
with the quasi-periodic condition
\begin{equation}
\label{bloch04}
 \varphi(x+1) = e^{i\theta}\varphi(x)\quad (x\in \mathbf{R}\setminus\mathbf{Z})
\end{equation}
for some $\theta \in \mathbf{R}$.
\begin{proposition}
\label{proposition_blochwave}
\begin{enumerate}
 \item 
For $\theta\in \mathbf{R}$,
there exists a non-trivial solution $\varphi$ to 
(\ref{bloch01})-(\ref{bloch03}) 
satisfying the Bloch wave condition (\ref{bloch04})
if and only if
\begin{equation}
 \label{bloch09}
D(k)=2\cos \theta.
\end{equation}

 \item When (\ref{bloch09}) holds,
a solution $\varphi(x)$
to (\ref{bloch01})-(\ref{bloch03}) and (\ref{bloch04}) is given by
(\ref{bloch05}) with the coefficients 
\begin{equation}
 \label{bloch10}
\begin{pmatrix}
 A_0 \cr B_0
\end{pmatrix}
=
\begin{pmatrix}
 -k^{-1}\sin k \cr
\cos k -e^{i\theta}
\end{pmatrix},\quad
\begin{pmatrix}
 A_j \cr B_j
\end{pmatrix}
=
e^{ij\theta}
\begin{pmatrix}
 A_0 \cr B_0
\end{pmatrix}
\quad 
(j\in \mathbf{Z}).
\end{equation}
\end{enumerate}
\end{proposition}

\begin{proof}
(i)
It is easy to see a non-trivial solution to 
(\ref{bloch01})-(\ref{bloch03}) and (\ref{bloch04})
exists if and only if $T(k)$ has an eigenvalue $e^{i\theta}$,
and the latter condition is equivalent to 
(\ref{bloch09}), since $\det T(k)=1$ and $\tr T(k)=D(k)$.
(ii) When (\ref{bloch09}) holds,
the vector ${}^t\!(A_0\ B_0)$ given in (\ref{bloch10})
is an eigenvector of $T(k)$ with the eigenvalue $e^{i\theta}$.
Thus the second equation in (\ref{bloch10})
follows from (\ref{bloch06}).
\end{proof}

Proposition \ref{proposition_blochwave}
and the Bloch theorem imply $\lambda\in \sigma(H)$
if and only if (\ref{bloch09}) holds
for some $\theta \in \mathbf{R}$, that is,
\begin{equation}
 \label{bloch11}
 -2 \leq D(\sqrt{\lambda}) \leq 2,
\end{equation}
as already stated in (\ref{kronig-penney3}).
For $\lambda<0$, we have
\begin{equation}
\label{bloch12}
 D(\sqrt{\lambda})
= 2 \cosh \sqrt{-\lambda} + V (\sqrt{-\lambda})^{-1}\sinh \sqrt{-\lambda}.
\end{equation}
If $V>0$,
the right hand side of (\ref{bloch12}) is larger than $2$
and (\ref{bloch11}) does not hold for $\lambda<0$.
Thus, there is no negative part in $\sigma(H)$.
If $V<0$, then some negative value $\lambda$ belongs to $\sigma(H)$,
and the corresponding $k$ is pure imaginary.
However, 
we concentrate on the high energy limit
in the present paper,
and the existence of the negative spectrum does not affect 
our argument.
So we sometimes assume $V>0$ in the sequel,
in order to simplify the notation.
In this case, the results for $V<0$ will be stated in the remark.

By an elementary inspection of the graph of $y=D(k)$,
we find the following properties.
\begin{proposition}
\label{proposition_discriminant}
Assume $V>0$. Then,
\begin{enumerate}
 \item $D(0)= 2 + V$ and $D(n \pi) = 2\cdot (-1)^n$ for $n=1,2,\ldots$.

 \item  The equation $D(k)=2 \cdot (-1)^n$ has a unique solution $k=k_n$ in 
the open interval $(n\pi,(n+1)\pi)$ for $n=0,1,2,\ldots$.

 \item The equation $D'(k)=0$ has a unique solution $k=l_n$ in
the open interval $(n\pi,(n+1)\pi)$ for $n=1,2,\ldots$,
and $n\pi<l_n<k_n$.

 \item For convenience, we put $l_0=0$.
Then, $D(k)$ is monotone decreasing on $[l_n,l_{n+1}]$ for even $n$,
and  monotone increasing on $[l_n,l_{n+1}]$ for odd $n$.
\end{enumerate}
\end{proposition}
{\bf Remark.}
When $V<0$, 
we denote the solution to $D(k)=2\cdot (-1)^n$ in $((n-1)\pi,n\pi)$ by $k_n$,
and the solution to $D'(k)=0$ in $((n-1)\pi,n\pi)$ by $l_n$,
for $n=2,3,\ldots$.
\begin{figure}[htbp]
 \begin{center}
\begin{tabular}{cc}
\begin{minipage}[c]{7cm}
  \includegraphics[width=7cm]{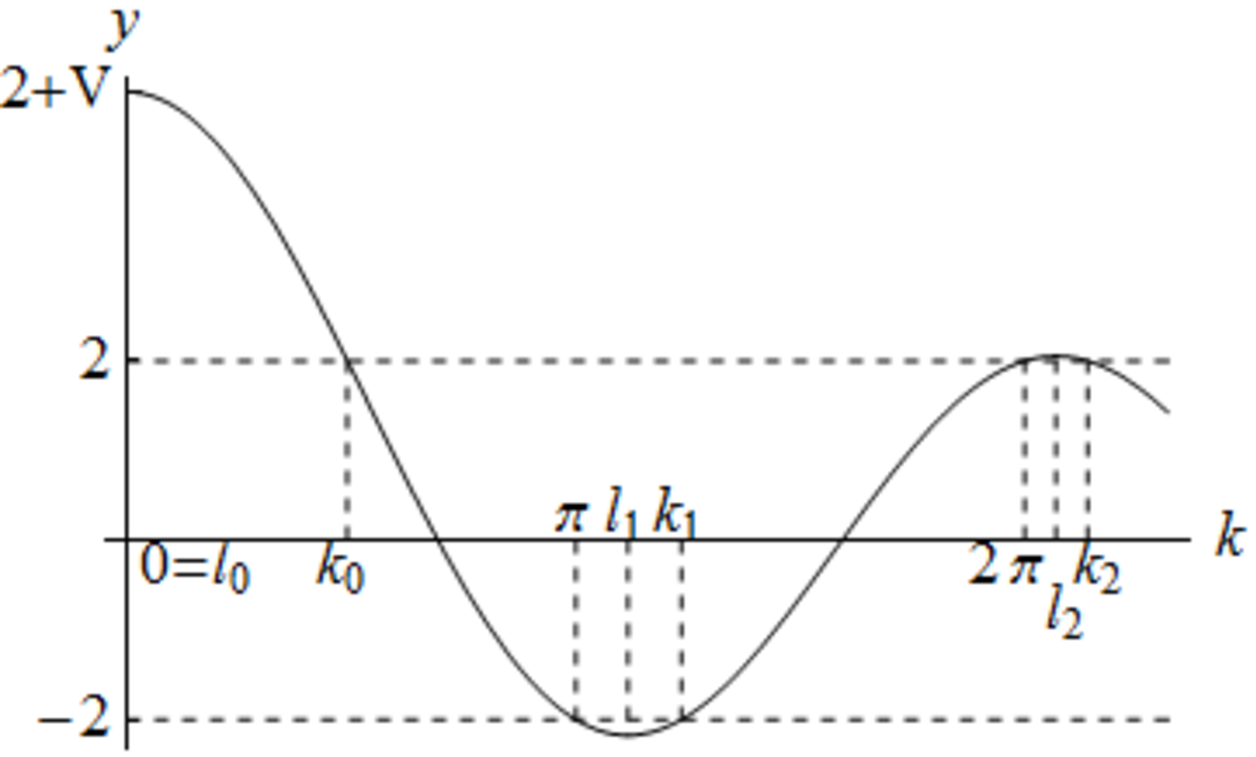}
 \caption{The graph of $y=D(k)$ when $V$ is positive.}
\end{minipage} 
& 
\begin{minipage}[c]{7cm}
  \includegraphics[width=7cm]{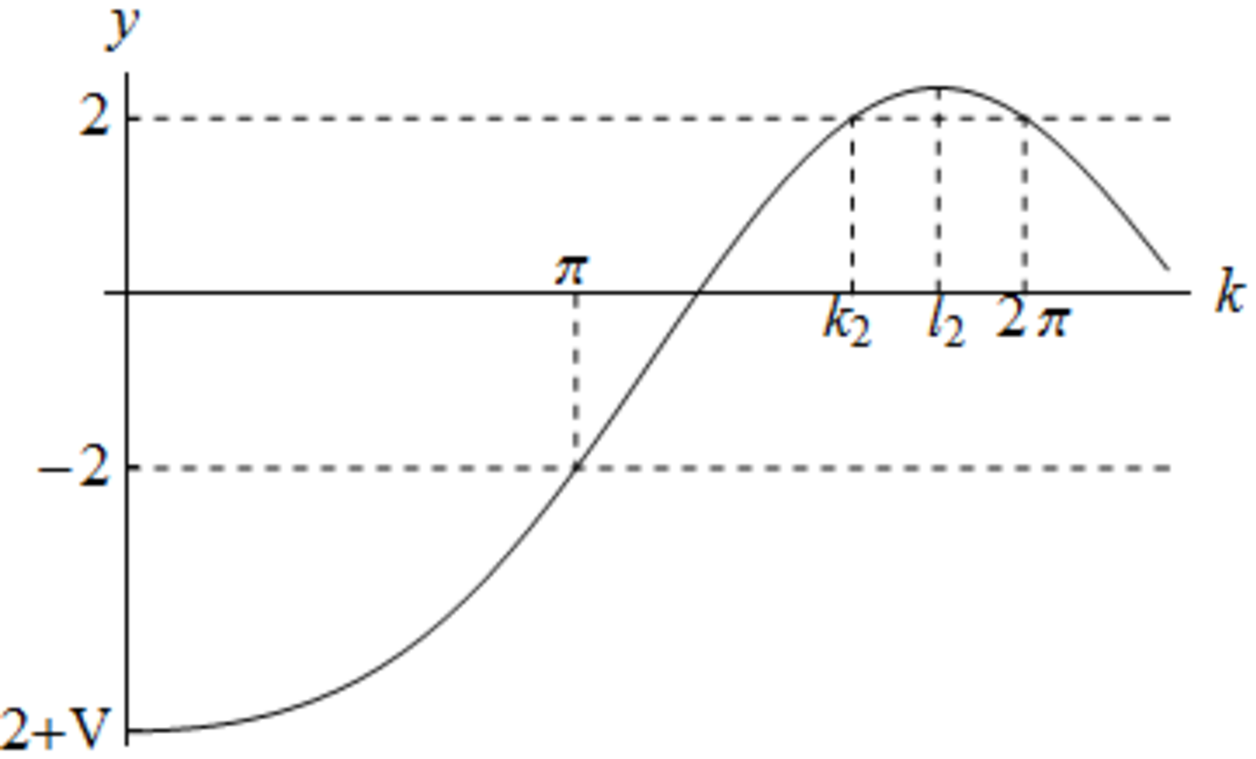}
 \caption{The graph of $y=D(k)$ when $V$ is negative.}
\end{minipage} 
\end{tabular}
 \end{center} 
\end{figure}

When $V>0$,
the spectrum of $H$ is given as
\begin{equation}
\label{bloch13}
 \sigma(H)=\bigcup_{n=1}I_n,\quad I_n=[k_{n-1}^2,(n\pi)^2],
\end{equation}
by (\ref{bloch11}) and Proposition \ref{proposition_discriminant}.
The closed interval $I_n$ is called the \textit{$n$-th band}.
By the expression (\ref{bloch13}),
the band gap $((n\pi)^2,k_n^2)$ is non-empty for every $n=1,2,\ldots$.
Proposition \ref{proposition_discriminant} also implies
the function $y=D(k)$ ($l_{n-1}\leq k \leq l_n$)
has the unique inverse function $k=D^{-1}(y)$.
Then the \textit{band function} $\lambda_n(\theta)$
is defined by
\begin{equation}
\label{band_function}
\lambda_n(\theta) =(k(\theta))^2,\quad
k(\theta)=D^{-1}(2\cos\theta)\quad
(k_{n-1}\leq k(\theta) \leq n\pi).
\end{equation}
By definition,
the band function $\lambda_n(\theta)$ is 
a real-analytic, periodic, and even function with respect to 
$\theta\in \mathbf{R}$.
The $n$-th band $I_n$ is the range of the band function $\lambda_n$.

Let $P_n$ be the spectral projection for the self-adjoint 
operator $H$ corresponding to the $n$-th band $I_n$.
The spectral theorem implies
\begin{displaymath}
 e^{-itH} = \textrm{s-}\lim_{N\to \infty} \sum_{n=1}^N P_n e^{-itH},
\end{displaymath}
where $\textrm{s-}\lim$ means the strong limit
in $L^2(\mathbf{R})$.
Let $K_{n,t}(x,y)$ be the integral kernel of the operator $P_n e^{-itH}$,
that is,
\begin{displaymath}
 P_n e^{-itH}u(x) = \textrm{s-}\lim_{N\to\infty}
\int_{-N}^N K_{n,t}(x,y) u(y) dy
\end{displaymath}
for $u\in L^2(\mathbf{R})$.
Let us calculate $K_{n,t}(x,y)$ more explicitly.
\begin{proposition}
\label{proposition_integral_kernel}
Assume $V>0$.
For $n=1,2,\ldots$ and $\theta\in \mathbf{R}$,
let $k=k(\theta)$ given by (\ref{band_function}).
Let $\varphi_{n,\theta}$ be the Bloch wave function
defined by (\ref{bloch05}) and (\ref{bloch10}) 
with $k=k(\theta)$.
Then, for any $x,y\in (0,1)$ and $j,m\in \mathbf{Z}$,
we have
\begin{equation}
 K_{n,t}(x+j,y+m)
=
\frac{1}{2\pi}
\int_{-\pi}^{\pi}
e^{-it k^2+ i (j-m) \theta}
\Phi_n(\theta,x,y)\cdot \frac{dk}{d\theta}
d\theta,\label{bloch16}
\end{equation}
\begin{eqnarray*}
\Phi_n(\theta,x,y)
&=&
\frac{\sin k}{\sin \theta}
\left(
\cos kx + \frac{V}{2k}\sin kx
\right)
\left(
\cos ky + \frac{V}{2k}\sin ky
\right)\\
&&\hspace{.5cm}+ i \sin k(x-y)
+
\frac{\sin \theta}{\sin k}\sin kx \sin ky,
\end{eqnarray*}
\begin{eqnarray}
 \frac{dk}{d\theta}
= \frac{-2\sin\theta}{D'(k)}
=
\frac{-2\sin\theta}{-2\sin k + V (k^{-1}\cos k - k^{-2}\sin k)}.
\label{bloch17.5}
\end{eqnarray}
\end{proposition}
{\bf Remark.}
When $V<0$, the same result holds for $n=2,3,\ldots$,
but the range of the function $k=k(\theta)$ is $[(n-1)\pi, k_n]$.

Before the proof, we prepare a lemma
about the Wronskian of the Bloch waves.
The Wronskian of two functions $\varphi$ and $\psi$ is defined as
\begin{displaymath}
 W(\varphi,\psi)=\varphi \psi'-\varphi'\psi.
\end{displaymath}

\begin{lemma}
\label{lemma_wronskian}
\begin{enumerate}
 \item 
For any two solutions $\varphi$ and $\psi$
to (\ref{bloch01})-(\ref{bloch03}),
the Wronskian $W(\varphi,\psi)$
is a constant function on $\mathbf{R}\setminus\mathbf{Z}$.

 \item 
Let $k=k(\theta)$ and $\varphi_{n,\theta}$
given in Proposition \ref{proposition_integral_kernel}.
Let $u_{n,\theta}$ be the normalized Bloch wave function
defined by
\begin{displaymath}
 u_{n,\theta}= \varphi_{n,\theta}/C_{n,\theta},\quad
C_{n,\theta} = \left(\int_0^1 |\varphi_{n,\theta}(x)|^2\,dx \right)^{1/2}.
\end{displaymath}
Then we have
\begin{equation}
\label{conjugate_u}
 \overline{u_{n,\theta}(x)}=u_{n,-\theta}(x),
\end{equation}
\begin{equation}
\label{wronskian}
 k \frac{d k}{d \theta}
= 
\frac{i}{2} W(u_{n,\theta}, u_{n,-\theta}),
\end{equation}
\end{enumerate}
\end{lemma}
\begin{proof}
(i) It is well-known that $W(\varphi,\psi)$ is constant on $(j,j+1)$
for every $j\in \mathbf{Z}$,
since $\varphi$ and $\psi$ are solutions to (\ref{bloch01}).
Moreover, (\ref{bloch02}) and (\ref{bloch03}) imply
\begin{eqnarray*}
 W(\varphi,\psi)(j+)
&=&\det
\begin{pmatrix}
 \varphi(j+) & \psi(j+)\cr
 \varphi'(j+) & \psi'(j+)\cr
\end{pmatrix}\\
&=&
\det
\begin{pmatrix}
 1& 0\cr
 V& 1
\end{pmatrix}
\det
\begin{pmatrix}
 \varphi(j-) & \psi(j-)\cr
 \varphi'(j-) & \psi'(j-)\cr
\end{pmatrix}
= W(\varphi,\psi)(j-)
\end{eqnarray*}
for every $j\in \mathbf{Z}$.
Thus $W(\varphi,\psi)$ is constant on $\mathbf{R}\setminus\mathbf{Z}$.

(ii)
The first statement (\ref{conjugate_u}) follows immediately
from the definition (\ref{bloch05}) and (\ref{bloch10}).
We introduce an auxiliary function
$v_{n,\theta}(x)$
by 
$u_{n,\theta}(x)=e^{ix\theta}v_{n,\theta}(x)$.
Then we have
\begin{equation}
 D(0)u_{n,\theta}= e^{ix\theta}D(\theta)v_{n,\theta},
\quad
D(0)= \frac{1}{i}\frac{d}{dx},
\quad
D(\theta) = \frac{1}{i}\frac{d}{dx}+ \theta.
\label{wr00}
\end{equation}
Since $\varphi = \varphi_{n,\theta}$ satisfies 
(\ref{bloch01})-(\ref{bloch03}) and (\ref{bloch04}),
it is easy to check
\begin{eqnarray}
&& 
D(\theta)^2 v_{n,\theta}(x) = 
k^2  v_{n,\theta}(x)\quad (x \in \mathbf{R}\setminus \mathbf{Z}),
\label{wr01}
\\
&& v_{n,\theta}(j+)=v_{n,\theta}(j-)\quad (j \in \mathbf{Z}),\label{wr02}\\
&& v_{n,\theta}'(j+)-v_{n,\theta}'(j-)=V v_{n,\theta}(j)\quad (j \in\mathbf{Z}),\label{wr03}\\
&& v_{n,\theta}(x+1)=v_{n,\theta}(x)
\quad (x \in \mathbf{R}\setminus \mathbf{Z}).\label{wr04}
\end{eqnarray}
In this proof, we denote $(\cdot,\cdot)$ the inner product in $L^2((0,1))$,
that is, $(u,v)=\int_0^1 \overline{u(x)}v(x)dx$.
Then we have from (\ref{wr01})
\begin{eqnarray}
(v_{n,\theta},v_{n,\theta})=1,\label{wr05}
\\
(v_{n,\theta}, D(\theta)^2 v_{n,\theta})=k^2.
\label{wr06}
\end{eqnarray}
By differentiating both sides of (\ref{wr05})
with respect to $\theta$, we have
\begin{equation}
(\partial_\theta v_{n,\theta}, v_{n,\theta})
+
(v_{n,\theta},\partial_\theta v_{n,\theta})
=0,
\label{wr07}
\end{equation}
where $\partial_\theta=\partial/\partial\theta$.
By differentiating both sides of (\ref{wr06})
with respect to $\theta$, we have
 \begin{equation}
(\partial_\theta v_{n,\theta}, D(\theta)^2 v_{n,\theta})
+
2
(v_{n,\theta}, D(\theta) v_{n,\theta})
+ 
(v_{n,\theta}, D(\theta)^2 \partial_\theta v_{n,\theta})
= 2k \partial_\theta k.
\label{wr08}
\end{equation}
By differentiating (\ref{wr02})-(\ref{wr04}) 
with respect to $\theta$,
we see that the derivative $\partial_\theta v_{n,\theta}$
also satisfies the same relations (\ref{wr02})-(\ref{wr04}). 
Then we have by integration by parts
\begin{equation}
\label{wr09}
(v_{n,\theta}, D(\theta)^2 \partial_\theta v_{n,\theta})
=
(D(\theta)^2 v_{n,\theta},  \partial_\theta v_{n,\theta}).
\end{equation}
By (\ref{wr01}), (\ref{wr07}) and (\ref{wr09}),
the first term and the third in the left hand side of (\ref{wr08})
cancel with each other. Thus we have
\begin{eqnarray*}
 k \partial_\theta k 
&=& (v_{n,\theta}, D(\theta) v_{n,\theta})\\
&=& \frac{1}{2}\left(
(v_{n,\theta}, D(\theta)v_{n,\theta})
+
(D(\theta) v_{n,\theta}, v_{n,\theta})\right)\\
&=& \frac{1}{2}\left(
(u_{n,\theta}, D(0)u_{n,\theta})
+
(D(0) u_{n,\theta}, u_{n,\theta})\right)\\
&=&
\frac{i}{2}
\int_0^1
\left(-
\overline{u_{n,\theta}(x)}u_{n,\theta}'(x)
+
\overline{u_{n,\theta}'(x)}u_{n,\theta}(x)
\right)dx\\
&=& \frac{i}{2}W(u_{n,\theta},u_{n,-\theta}).
\end{eqnarray*}
Here we use 
(\ref{wr02})-(\ref{wr04}) in the second equality,
(\ref{wr00}) in the third,
and $\overline{u_{n,\theta}}=u_{n,-\theta}$ in the last.
\end{proof}

\begin{proof}[Proof of Proposition \ref{proposition_integral_kernel}]
First, the Floquet-Bloch theory tells us
\begin{eqnarray}
K_{n,t}(x,y)
=\frac{1}{2\pi}
\int_{-\pi}^\pi e^{-it\lambda_n(\theta)}u_{n,\theta}(x)
\overline{u_{n,\theta}(y)}\,d\theta
\label{bloch14}
\end{eqnarray}
for any $x,y\in \mathbf{R}$.
Actually, $u_{n,\theta}$ is the normalized eigenfunction
of $H_\theta=H|_{\mathscr{H}_\theta}$, where 
$\mathscr{H}_\theta$ is the Hilbert space defined by
\begin{equation}
\mathscr{H}_\theta = \{u\in L^2_{\rm loc}(\mathbf{R})
:
u(x+1)=e^{i\theta}u(x)
\},
\quad
\|u\|_{\mathscr{H}_\theta}^2=\int_0^1|u(x)|^2\,dx.
\end{equation}
Since the whole operator $H$ has the direct integral decomposition
$H = \int_{[-\pi,\pi)}^\oplus H_\theta\, d\theta/(2\pi)$,
the formula (\ref{bloch14}) follows from the eigenfunction expansion
for $H_\theta$
(for the detail, see e.g.\ Reed--Simon \cite{ReSi4}).

Let $x,y\in (0,1)$ and $j,m\in \mathbf{Z}$.
Since $\lambda_n(\theta)=k(\theta)^2$,
$u_{n,\theta}(x+j)=e^{ij\theta}u_{n,\theta}(x)$,
and $\overline{u_{n,\theta}(y)}=u_{n,-\theta}(y)$,
we have from (\ref{bloch14})
\begin{eqnarray}
 K_{n,t}(x+j,y+m)
= 
\frac{1}{2\pi}
\int_{-\pi}^{\pi}
e^{-itk^2+i(j-m)\theta}
u_{n,\theta}(x)u_{n,-\theta}(y)
\left(\frac{dk}{d\theta}\right)^{-1}
\cdot\frac{dk}{d\theta}
d\theta.
\label{bloch18}
\end{eqnarray}
Moreover, we have by (\ref{wronskian}) 
\begin{eqnarray}
 \left(\frac{dk}{d\theta}\right)^{-1}
u_{n,\theta}(x)u_{n,-\theta}(y)
=\frac{2k}{i}\frac{u_{n,\theta}(x)u_{n,-\theta}(y)}{W(u_{n,\theta},u_{n,-\theta})}
=\frac{2k}{i}\frac{\varphi_{n,\theta}(x)\varphi_{n,-\theta}(y)}{W(\varphi_{n,\theta},\varphi_{n,-\theta})}.\label{bloch19}
\end{eqnarray}
The Wronskian 
$W(\varphi_{n,\theta},\varphi_{n,-\theta})$
on the interval $(0,1)$ is calculated as follows.
\begin{eqnarray}
&& W(\varphi_{n,\theta},\varphi_{n,-\theta})\nonumber\\
&=&
\det
\begin{pmatrix}
 A_0(\theta)\cos k x +  B_0(\theta)k^{-1}\sin k x &
 A_0(-\theta)\cos k x +  B_0(-\theta)k^{-1}\sin k x \\
 -k A_0(\theta)\sin k x + B_0(\theta) \cos k x &
 -k A_0(-\theta)\sin k x +  B_0(-\theta)\cos k x 
\end{pmatrix}\nonumber\\
&=&
\det
\begin{pmatrix}
 \cos kx & k^{-1}\sin kx\cr
 -k\sin kx & \cos kx
\end{pmatrix}
\det
\begin{pmatrix}
 A_0(\theta) & A_0(-\theta)\cr
B_0(\theta) & B_0(-\theta)
\end{pmatrix}
\nonumber\\
&=&
A_0(\theta)B_0(-\theta)
-
A_0(-\theta) B_0(\theta)\nonumber\\
&=&
-k^{-1}\sin k (\cos k - e^{-i\theta})
+k^{-1}\sin k (\cos k - e^{i\theta})\nonumber\\
&=&- 2i k^{-1}\sin k \sin\theta.
\label{bloch20}
\end{eqnarray}
By (\ref{bloch08}) and (\ref{bloch09}), we have
\begin{displaymath}
e^{i\theta} = \cos k + \frac{V}{2}k^{-1}\sin k + i \sin\theta.
\end{displaymath}
This equality and (\ref{bloch10}) implies
for $0<x<1$
\begin{eqnarray}
\varphi_{n,\theta}(x)
&=& A_0(\theta)\cos kx + B_0(\theta)\sin kx\nonumber\\
&=& -k^{-1}\sin k \left(\cos kx + \frac{V}{2k}\sin kx\right)
- ik^{-1} \sin\theta \sin kx.
\label{bloch21}
\end{eqnarray}
Substituting (\ref{bloch20}) and (\ref{bloch21}) into (\ref{bloch19}), we have
\begin{eqnarray*}
&& \left(\frac{dk}{d\theta}\right)^{-1}
u_{n,\theta}(x)u_{n,-\theta}(y)\\
&=&
\frac{k^2}{\sin k\sin\theta} \varphi_{n,\theta}(x)\varphi_{n,-\theta}(y)\\
&=&
\frac{1}{\sin k\sin\theta} 
\left(
\sin k \left(\cos kx + \frac{V}{2k}\sin kx\right)
+i \sin\theta \sin kx
\right)\\
&&\hspace{2cm}\cdot
\left(
\sin k \left(\cos ky + \frac{V}{2k}\sin ky\right)
-i \sin\theta \sin ky
\right)\\
&=& \Phi_n(\theta,x,y).
 \end{eqnarray*}
Substituting this equality into (\ref{bloch18}),
we have (\ref{bloch16}).
The derivative (\ref{bloch17.5}) is obtained
by differentiating $2\cos\theta=D(k)=2\cos k+ Vk^{-1}\sin k$.
\end{proof}
{\bf Remark.}
Note that 
$dk/d\theta$ is positive for $0< \theta <\pi$ for odd $n$,
and negative for even $n$.
Substituting
(\ref{bloch19}) into (\ref{bloch18})
and making the change of variable $\lambda=(k(\theta))^2$,
we have for $V>0$
\begin{eqnarray}
 K_{n,t}(x,y)&=&\frac{1}{i\pi}
\int_{-\pi}^\pi 
e^{-it k^2}
\frac{\varphi_{n,\theta}(x)\varphi_{n,-\theta}(y)}
{W(\varphi_{n,\theta},\varphi_{n,-\theta})}
k
\cdot\frac{dk}{d\theta}
d\theta\nonumber\\
&=&
\frac{(-1)^{n-1}}{\pi}
\int_{k_{n-1}^2}^{(n\pi)^2}
e^{-it \lambda}
\Im\left[
\frac{\varphi_{n,\theta}(x)\varphi_{n,-\theta}(y)}
{W(\varphi_{n,\theta},\varphi_{n,-\theta})}
\right]
d\lambda.
\label{another_representation}
\end{eqnarray}
The formula (\ref{another_representation}) can be deduced in another way.
According to the Stone formula (\cite[Theorem VII.13]{ReSi1}),
the spectral measure $E_H$ for the operator $H$ 
is represented as
\begin{equation}
\label{bloch22}
 \frac{d E_H}{d\lambda} = \frac{1}{\pi}\Im R_H(\lambda+ i 0),
\end{equation}
where 
$R_H(\lambda+0)=\lim_{\epsilon\to +0}R(\lambda+i\epsilon)$ 
is the boundary value of the resolvent
$R(\lambda)=(H-\lambda)^{-1}$.
The integral kernel of the resolvent operator $R(\lambda)$ 
for $\lambda\in \rho(H)$ (the resolvent set of $H$) is given as
\begin{equation}
\label{bloch23}
 R(\lambda)(x,y)
=
-\frac{\varphi_+^\lambda(x \vee y)\varphi_-^\lambda(x\wedge y)}{W(\varphi_+^\lambda,\varphi_-^\lambda)},
\end{equation}
where 
$x\vee y=\max(x,y)$,
$x\wedge y=\min(x,y)$,
and $\varphi_\pm^\lambda$ is the solution to
(\ref{bloch01})-(\ref{bloch03}) decaying exponentially
as $x\to \pm \infty$.
The two functions $\varphi_\pm^\lambda$ are given by (\ref{bloch05})
and (\ref{bloch10}) with $e^{i\theta}$ replaced by
the solutions to 
\begin{displaymath}
 z^2 - D(\sqrt{\lambda})z + 1=0,
\end{displaymath}
so that $|z|^{\pm 1}<1$.
By choosing the solution $z$ appropriately,
(\ref{bloch22}) and (\ref{bloch23}) give 
the formula (\ref{another_representation}).

\section{Analysis of band functions}
In this section, 
we analyze the band function $\lambda_n(\theta)$, 
which is explicitly given by
\begin{equation}
\label{ba01}
 \lambda_n(\theta) = (k(\theta))^2, \quad 
k(\theta)=D^{-1}(2\cos \theta),
\end{equation}
especially its asymptotics as $n\to \infty$.
Here $k=D^{-1}(y)$ is the inverse function 
of 
\begin{displaymath}
 y=D(k)=2\cos k + V\frac{\sin k}{k}:
[l_{n-1},l_n]\to [y_{n-1},y_n],
\end{displaymath}
$l_n$ is given in Proposition \ref{proposition_discriminant}
and its remark,
and we put $y_n=D(l_n)$.
By the formula (\ref{ba01}),
we can draw the graphs of $\lambda_n(\theta)$,
\footnote{
The inverse correspondence $\theta = \arccos(D(k)/2)$
is useful for the numerical calculation.
}
which are illustrated in Figure \ref{figure_band_function} in the introduction.
From Figure \ref{figure_band_function},
we notice that
$\lambda=\lambda_n(\theta)$ 
is similar to the parabola $\lambda=\theta^2$
on the interval $[(n-1)\pi, n\pi]$,
except near the edge points.
From this reason, we 
mainly consider the function $\lambda_n(\theta)$ 
on the interval $[(n-1)\pi, n\pi]$,
thereby we can simplify some formulas given below.
Figure \ref{figure_band_function} also suggests us
it is better to analyze $\lambda_n(\theta)$ when its value is near 
the band edge,
and near the band center, separately.

\subsection{Explicit formulas for derivatives}

Our goal is to give an asymptotic bound for the oscillatory integral 
(\ref{bloch16}),
 by using Lemma \ref{lemma_stein} in Section 4.
Then we need lower bounds for the derivatives 
of $\lambda_n(\theta)$ up to the third order,
which are calculated explicitly as follows.
\begin{eqnarray}
 k'(\theta) &=& -D'(k)^{-1}\cdot 2\sin \theta,\label{der1}\\
k''(\theta) &=&
 (D^{-1})''(2\cos\theta)\cdot (2\sin\theta)^2
-
 (D^{-1})'(2\cos\theta)\cdot 2\cos\theta\nonumber\\
&=& - D'(k)^{-3}D''(k)(4-D(k)^2)
-
D'(k)^{-1}D(k),\label{der2}\\
k^{(3)}(\theta)&=&
-(D^{-1})^{(3)}(2\cos\theta)\cdot (2\sin\theta)^3
+ 3\cdot (D^{-1})^{''}(2\cos\theta)\cdot 2\sin\theta \cdot 2\cos\theta
\nonumber\\
&&+(D^{-1})^{'}(2\cos\theta)\cdot 2\sin\theta,\nonumber\\
&=&
\Bigl(
\left(
-3D'(k)^{-5}D''(k)^2 + D'(k)^{-4}D^{(3)}(k)
\right)(4-D(k)^2)\nonumber\\
&&\quad -3D'(k)^{-3}D(k)D''(k) + D'(k)^{-1}
\Bigr)
\cdot 2\sin \theta,
\label{der3}
\\
\lambda_n'(\theta)&=& 
2k(\theta) k'(\theta),\label{der4}\\
\lambda_n''(\theta)&=&
2k'(\theta)^2 + 2k(\theta) k''(\theta),\label{der5}\\
\lambda_n^{(3)}(\theta)&=&
6k'(\theta)k''(\theta) + 2k(\theta) k^{(3)}(\theta),\label{der6}
\end{eqnarray}
where we used the formulas
\begin{eqnarray}
\label{dkformula1}
&& 2\cos\theta=D(k),\ (2\sin\theta)^2=4-D(k)^2,\\
&& 
\label{dkformula2}
(D^{-1})'(D(k))=D'(k)^{-1},\ 
(D^{-1})''(D(k))=-D'(k)^{-3}D''(k),\\
&&(D^{-1})^{(3)}(D(k))=3D'(k)^{-5}D''(k)^2-D'(k)^{-4}D^{(3)}(k).
\nonumber
\end{eqnarray}
The derivatives of $D(k)$ are given as follows.
\begin{eqnarray}
 D(k)&=& 2\cos k + V k^{-1}\sin k,\label{d0k}\\
 D'(k)&=& -2\sin k + V(k^{-1}\cos k-k^{-2}\sin k),
\label{d1k}\\
 D''(k)&=& -2\cos k + V(-k^{-1}\sin k-2k^{-2}\cos k + 2k^{-3}\sin k),
\label{d2k}\\
 D^{(3)}(k)&=&  2\sin k + V(-k^{-1}\cos k+3k^{-2}\sin k + 6k^{-3}\cos k - 6k^{-4}\sin k),
\label{d3k}\\
 D^{(4)}(k)&=&  2\cos k + V(k^{-1}\sin k+4k^{-2}\cos k -12k^{-3}\sin k - 24k^{-4}\cos k
\nonumber\\
&&
+ 24 k^{-5} \sin k),\label{d4k}\\
 D^{(5)}(k)&=&  -2\sin k + V(k^{-1}\cos k-5k^{-2}\sin k -20k^{-3}\cos k + 60k^{-4}\sin k\nonumber\\
&&
+ 120 k^{-5} \cos k-120k^{-6}\sin k),
\label{d5k}
\\
 D^{(6)}(k)&=&  -2\cos k + V(-k^{-1}\sin k-6k^{-2}\cos k +30k^{-3}\sin k + 120k^{-4}\cos k\nonumber\\
&&
- 360 k^{-5} \sin k -720k^{-6}\cos k + 720 k^{-7}\sin k).
\label{d6k}
\end{eqnarray}

By the formulas (\ref{der1})-(\ref{der6}),
we can write $\lambda_n^{(j)}(\theta)$ ($j=1,2,3$)
as functions of $k$, which is useful for numerical calculation.
The graphs of $\lambda_n'(\theta)$ and $\lambda_n''(\theta)$
on the interval $(n-1)\pi\leq \theta\leq n \pi$ 
($n=1,2,3,4,5$)
are illustrated 
in Figure \ref{figure_lambda1} and 
Figure \ref{figure_lambda2},
respectively.
From Figure \ref{figure_lambda2},
we see that the solution $\theta_0$ to $\lambda_n''(\theta)=0$ exists
nearby $\theta=n\pi$.
Later we give the asymptotics of $\theta_0$ as $n\to \infty$ in 
Proposition \ref{proposition_inflection}.
\begin{figure}[htbp]
 \begin{center}
  \begin{tabular}{cc}
\begin{minipage}[c]{7cm}
 \includegraphics[width=7cm]{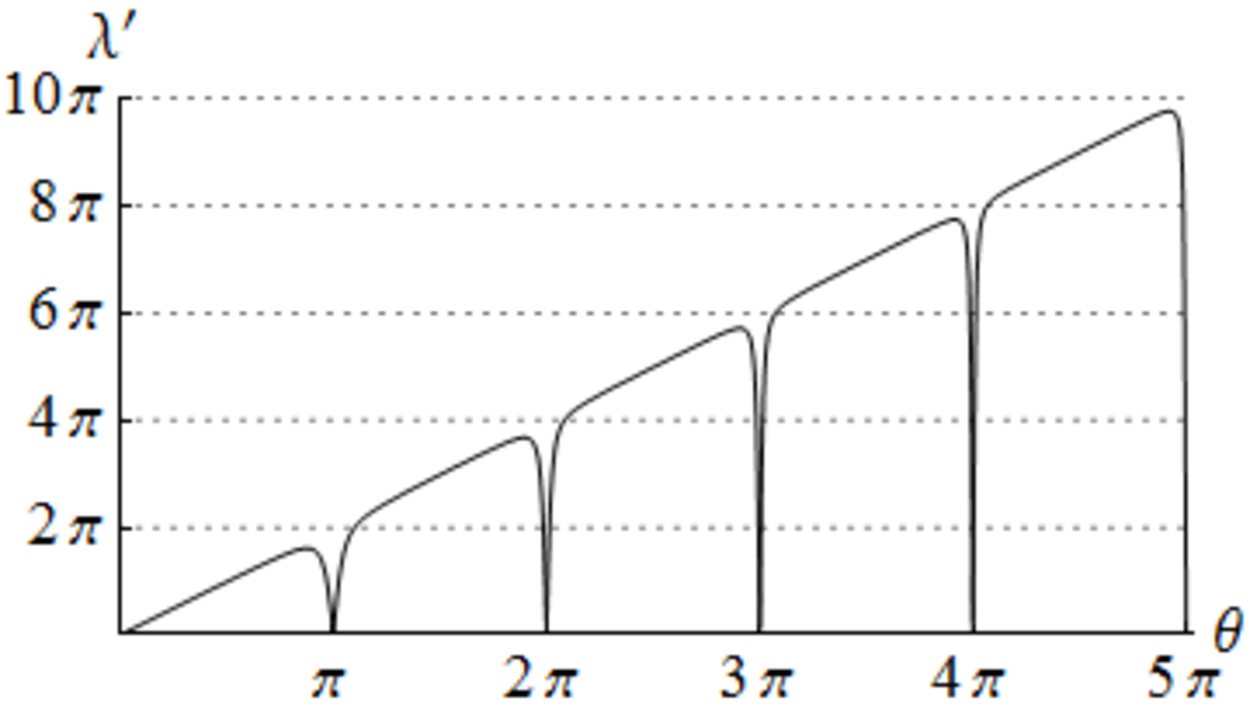}
\caption{The graphs of $\lambda_n'(\theta)$ 
on $[(n-1)\pi,  n\pi]$
($n=1,2,3,4,5$) for $V>0$.}
\label{figure_lambda1}
\end{minipage}
   & 
\begin{minipage}[c]{7cm}
 \includegraphics[width=7cm]{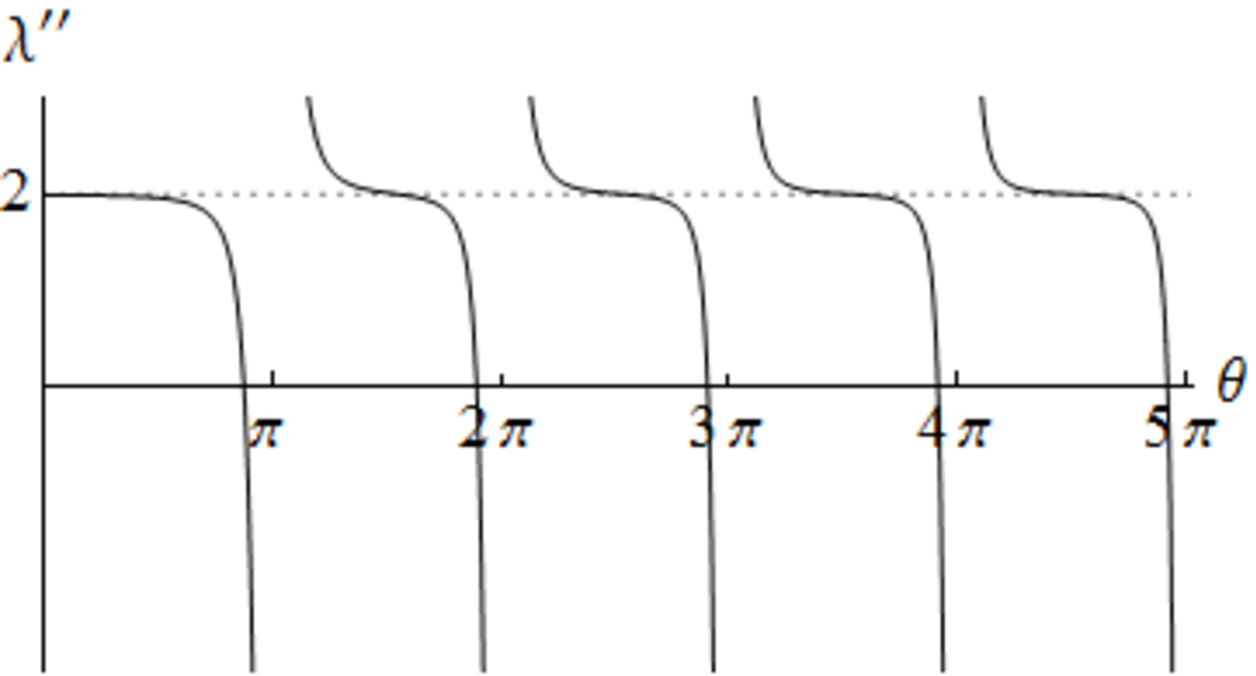}
\caption{The graphs of $\lambda_n''(\theta)$ 
on $[(n-1)\pi,  n\pi]$
($n=1,2,3,4,5$) for $V>0$.}
\label{figure_lambda2}
\end{minipage}
  \end{tabular}
 \end{center}
\end{figure}

Though the formulas (\ref{der1})-(\ref{der6}) are explicit,
it is still not easy to obtain the precise lower bound 
for the derivatives of $\lambda_n(\theta)$,
especially when $\lambda_n(\theta)$ is near the band edge.
For this reason, we employ the \textit{Puiseux expansion}
of the inverse function $D^{-1}(y)$, 
which makes our analysis clear.
This kind of expansion is studied in the classical work by 
Kohn \cite{Kohn}.

\subsection{Asymptotics of $k_n$ and $l_n$}
First let us analyze the asymptotics of $k_n$ and $l_n$
given in Proposition \ref{proposition_discriminant}.
A related result is written in 
Albeverio et.\ al.\ \cite[Theorem 2.3.3]{Alb2nd}.
\begin{proposition}
 \label{proposition_knln}
Let $k_n$ and $l_n$
as in Proposition \ref{proposition_discriminant}
and its remark.
Then,
\begin{equation}
 \label{asymptotic_kn}
k_n = 
n\pi 
+ V(n\pi)^{-1}
-\left(V^2+\frac{V^3}{12}\right)(n\pi)^{-3}
+ O (n^{-5})\qquad (n\to \infty),
\end{equation}
\begin{equation}
 \label{asymptotic_ln}
l_n = 
n\pi 
+ \frac{V}{2}(n\pi)^{-1}
-\left(\frac{V^2}{2}+\frac{V^3}{24}\right)(n\pi)^{-3}
+ O (n^{-5})\qquad (n\to \infty).
\end{equation}
\end{proposition}
\begin{proof}
We assume $V>0$ for simplicity
(in the case $V<0$, we only need to change the sign of $h$ given below).
First we prove (\ref{asymptotic_kn}).
The number $k_n$ is the solution to 
\begin{equation}
\label{ba02}
 D(k)=2\cos k + V\frac{\sin k}{k}=2(-1)^n\quad
(n\pi<k<(n+1)\pi).
\end{equation}
Put $k = n\pi + h$. Then (\ref{ba02}) is equivalent to
\begin{equation}
\label{ba03}
h = f(h),\quad f(h)=\arcsin\left(\frac{1+\cos h}{2} \frac{V}{n\pi + h} \right).
\end{equation}
Then, for sufficiently large $n$,
it is easy to see that the solution to (\ref{ba03})
is the limit of the sequence $\{h_j\}_{j=0}$ given by%
\footnote{For the rigorous proof, we apply the contraction mapping theorem
in the following form:
`Let $I=[a,b]$, $0\in I$, $f\in C^1(I)$ and assume 
$f(I)\subset I$ and $\|f'\|_{L^\infty(I)}<1$.
Then $f$ has the unique fixed point in $I$ 
which is the limit of the sequence (\ref{ba04})'.
If we take $f$ as in (\ref{ba03}) and $I=[0,2V/(n\pi)]$,
we can apply the contraction mapping theorem for sufficiently large $n$.
}
\begin{equation}
\label{ba04}
 h_0=0,\quad h_j=f(h_{j-1})\quad (j=1,2,\ldots).
\end{equation}
By a simple calculation, we have
\begin{equation} 
\label{ba05}
f(h)
= 
\frac{V}{n\pi}-\frac{V}{(n\pi)^2}h
-\frac{V}{4n\pi}h^2 + \frac{1}{6}\frac{V^3}{(n\pi)^3}+ O(n^{-5})\quad
\mbox{for } h=O(n^{-1}).
\end{equation}
Then twice substitution of (\ref{ba05}) into (\ref{ba04})
gives the formula (\ref{asymptotic_kn})
(three times substitution gives the same formula).

Next we prove (\ref{asymptotic_ln}).
Put $k=n\pi + h$. Then the defining equation of $l_n$
\begin{eqnarray*}
 D'(k)=-2\sin k + V(k^{-1}\cos k - k^{-2}\sin k)=0
\end{eqnarray*}
is equivalent to
\begin{eqnarray*}
 h=g(h), \quad g(h)=\arctan\left(\frac{Vk}{2k^2+V}\right),\quad
k=n\pi + h.
\end{eqnarray*}
Then we can obtain the formula
(\ref{asymptotic_ln}) by using the following expansion recursively.
\begin{displaymath}
 g(h)=
\frac{V}{2n\pi}
-
\frac{V^2}{4(n\pi)^3}
-\frac{V}{2(n\pi)^2}h
-\frac{1}{24}\frac{V^3}{(n\pi)^3}
+ O(n^{-5})
\quad
\mbox{for } h=O(n^{-1}).
\end{displaymath}
\end{proof}

\subsection{Analysis on $\lambda_n(\theta)$ near the band edge} 
In order to calculate the Puiseux expansion of $k=D^{-1}(y)$,
first we calculate the Taylor expansion of $y=D(k)$.

\begin{proposition}
 \label{proopsition_taylor_dk}
The Taylor expansion of $y=D(k)$
near $k=l_n$ is given as follows.
\begin{equation}
 \label{taylor_dk}
y=D(k) = y_n + \sum_{m=2}^\infty d_m (k-l_n)^m,\qquad
y_n = D(l_n),\quad
d_m = \frac{D^{(m)}(l_n)}{m!},
\end{equation}
\vspace{-4mm}
\begin{eqnarray}
\label{yn}
y_n &=& (-1)^n\left[2 + \frac{V^2}{4}(n\pi)^{-2}+O(n^{-4})\right],\\
\label{d2}
d_2 &=& (-1)^n\left[-1 - \left(V+\frac{V^2}{8}\right)(n\pi)^{-2}+O(n^{-4})\right],\\
\label{d3}
d_3 &=& (-1)^n\left[ \left(V+\frac{V^2}{6}\right)(n\pi)^{-3}+O(n^{-5})\right],\\
\label{d4}
d_4 &=& (-1)^n\left[\frac{1}{12}+ \left(\frac{V}{6}+\frac{V^2}{96}\right)(n\pi)^{-2}+O(n^{-4})\right],\\
\label{d5}
d_5 &=& (-1)^n\left[ -\left(\frac{V}{6}+\frac{V^2}{60}\right)(n\pi)^{-3}+O(n^{-5})\right],
\\
\label{d6}
d_6 &=& (-1)^n\left[-\frac{1}{360}- \left(\frac{V}{120}+\frac{V^2}{2880}\right)(n\pi)^{-2}+O(n^{-4})\right].
\end{eqnarray}
\end{proposition}
\begin{proof}
Notice that $d_1=D'(l_n)=0$ by definition.
The results (\ref{yn})-(\ref{d5}) 
can be obtained by substituting the following formulas
into (\ref{d0k})-(\ref{d6k}).
From (\ref{asymptotic_ln}), we have 
\begin{eqnarray*}
(-1)^n \cos l_n 
&=& 1 -\frac{V^2}{8}(n\pi)^{-2}
+ O(n^{-4}),\\
(-1)^n \sin l_n 
&=& 
\frac{V}{2}(n\pi)^{-1}
-\left(\frac{V^2}{2}+\frac{V^3}{16}\right)(n\pi)^{-3}
+ O(n^{-5}),\\
l_n^{-1}
&=& (n\pi)^{-1}- \frac{V}{2}(n\pi)^{-3} + O(n^{-5}),\\
l_n^{-j}
&=& (n\pi)^{-j} + O(n^{-(j+2)})\quad (j \geq 2).
\end{eqnarray*}
\end{proof}

Next we shall calculate the Puiseux expansion of $D^{-1}(y)$ near
$y=y_n$.
Since $k=k_n$ is the zero of order 2 of $D(k)$, 
$D^{-1}(y)$ is a double-valued function 
with respect to the complex variable $y$.
\begin{proposition}
\label{proposition_puiseux_k}
The Puiseux expansion of $k=D^{-1}(y)$ near $y=y_n$
is written as follows.
\begin{eqnarray}
\label{puiseux_k} 
k &=& l_n + \sum_{p=1}^\infty e_p 
h^{p/2},\quad
h = (-1)^n(y_n-y),\\
\label{e1}
 e_1 &=& 1 - \left(\frac{V}{2} + \frac{V^2}{16}\right)(n\pi)^{-2}+O(n^{-4}),\\
\label{e2}
e_2 &=& \left(\frac{V}{2} + \frac{V^2}{12} \right)(n\pi)^{-3}+O(n^{-5}),\\
\label{e3}
e_3 &=& \frac{1}{24} 
- \left(\frac{V}{48} + \frac{V^2}{128}\right)(n\pi)^{-2}+O(n^{-4}),\\
\label{e4}
e_4 &=& \left(\frac{V}{24}+ \frac{V^2}{80}\right)(n\pi)^{-3}+O(n^{-5}),\\
\label{e5}
e_5 &=& \frac{3}{640} -
\left(\frac{3V}{1280}+\frac{3V^2}{2048}\right)(n\pi)^{-2}+
O(n^{-4}).
\end{eqnarray}
\end{proposition}
\begin{proof}
Put $h=(-1)^n(y_n-y)$ and $c_m=(-1)^{n+1}d_m$, where $d_m$ are 
the coefficients in (\ref{taylor_dk}).
Then the Taylor expansion (\ref{taylor_dk}) is rewritten as
\begin{equation}
\label{pk01}
 h=\sum_{m=2}^\infty c_m (k-l_n)^m.
\end{equation}
Substituting the Puiseux expansion (\ref{puiseux_k}) into 
(\ref{pk01}), 
and comparing the coefficient of each power $h^{p/2}$,
we have
\begin{eqnarray*}
&&c_2 e_1^2 = 1,\\
&&2 c_2e_1e_2 +  c_3e_1^3=0,\\
&&c_2(2 e_1 e_3 +  e_2^2)+ 3c_3e_1^2e_2 + c_4 e_1^4=0,\\
&&c_2 (2 e_1e_4 + 2e_2e_3)+c_3(3e_1^2e_3+3e_1e_2^2)+4c_4e_1^3e_2+c_5e_1^5
=0,\\
&&c_2(2e_1e_5+2e_2e_4+e_3^2) 
  +c_3(6e_1e_2e_3+3e_1^2e_4+e_2^3)
+c_4 (6e_1^2e_2^2 + 4e_1^3e_3)\\
&&\hspace{9cm}+5c_5e_1^4e_2
+c_6e_1^6 =0.
\end{eqnarray*}
Since the expansions of $c_m=(-1)^{n+1}d_m$ are obtained 
from (\ref{d2})-(\ref{d6}),
we can calculate the coefficients (\ref{e1})-(\ref{e5})
by solving the above equations.
\end{proof}
{\bf Remark.} 1.
The branch points of $k=D^{-1}(y)$ are $y=y_n$,
and $y_n$ is nearby $2\cdot (-1)^n$ by (\ref{yn}).
Thus the radius of convergence of the Puiseux expansion
(\ref{puiseux_k}) is about $4$, for sufficiently large $n$.
This fact justifies the calculus in the sequel.

2. In the case $V=0$, the expansion (\ref{puiseux_k}) coincides with 
the Puiseux expansion
\begin{displaymath}
\arccos \frac{y}{2}=n\pi + (2\mp y)^{1/2}+\frac{1}{24}(2\mp y)^{3/2}+\frac{3}{640}(2\mp y)^{5/2}+\cdots\quad
(\mbox{near }y=\pm 2),
\end{displaymath}
where $n$ is any integer satisfying $(-1)^n=\pm 1$.
\vspace{2mm}

Substituting $y=2\cos\theta$ into (\ref{puiseux_k}),
we obtain the function $k=k(\theta)=D^{-1}(2\cos\theta)$.
Since the Puiseux expansion (\ref{puiseux_k})
has two branches,
we obtain two functions
\begin{eqnarray}
 \label{ba07}
&& k_\pm (\theta)=
l_n \pm e_1 h^{1/2} + e_2 h \pm e_3 h^{3/2} + e_4 h^2 \pm e_5 h^{5/2} + \cdots,
\\
&& h=(-1)^n(y_n-2\cos\theta).
 \label{ba08}
\end{eqnarray}
Notice that $h$ is positive for real $\theta$, 
so $h^{p/2}$ in (\ref{ba07}) is uniquely defined 
as a positive function.
For notational simplicity, we put
\begin{equation}
 \label{zndef}
z_n = (-1)^n y_n,\quad 
\theta=n\pi+\tau.
\end{equation}
Then we have the following formulas, useful for the calculation below.
\begin{eqnarray}
\label{formula1}
&& z_n = 2 + \frac{V^2}{4}(n\pi)^{-2} + O(n^{-4}),\\
\label{formula2}
&& h = z_n - 2 \cos \tau = \tau^2 + O(\tau^4)+
\frac{V^2}{4}(n\pi)^{-2} + O(n^{-4}),\\
\label{formula3}
&& \frac{dh}{d\tau}=2\sin \tau,\\
\label{formula4}
&&(2\sin \tau)^2=4 - (2\cos\tau)^2 = (4-z_n^2)+2z_n h -h^2,\\
\label{formula5}
&& 4-z_n^2 = -V^2(n\pi)^{-2} + O(n^{-4}).
\end{eqnarray}
Then we obtain expansions of 
the band functions and their derivatives near the band edge,
as follows.
\begin{proposition}
 \label{proposition_puiseux_lambda}
Let $k_\pm$, $h$ and $\tau$ as in  (\ref{ba07})-(\ref{zndef}),
and put $\lambda_\pm(\theta)= k_\pm(\theta)^2$.
Then the following expansions hold near $\theta=n\pi$
(or $\tau=0$).

\begin{enumerate}
 \item 
\begin{displaymath}
\lambda_\pm =l_n^2  
\pm \lambda_{0,1} h^{1/2}+ \lambda_{0,2} h
\pm \lambda_{0,3} h^{3/2}+ \lambda_{0,4} h^2
\pm \lambda_{0,5} h^{5/2} + O(h^3),
\end{displaymath}
\begin{eqnarray*}
 l_n^2 &=&
(n\pi)^2 + V -\left(\frac{3V^2}{4} + \frac{V^3}{12}\right)(n\pi)^{-2}+O(n^{-4}),
\\
 \lambda_{0,1}
&=&
2 n\pi -\frac{V^2}{8}(n\pi)^{-1}+O(n^{-3}),
\\
 \lambda_{0,2} 
&=&
1
+\frac{V^2}{24}(n\pi)^{-2}+O(n^{-4}),
\\
 \lambda_{0,3} &=&
\frac{1}{12}n\pi
-\frac{V^2}{64}(n\pi)^{-1}+O(n^{-3}),
\\
 \lambda_{0,4}&=&
\frac{1}{12}
+\frac{V^2}{240}(n\pi)^{-2}+O(n^{-4}),
\\
 \lambda_{0,5}
&=&\frac{3}{320}n\pi -\frac{3V^2}{1024}(n\pi)^{-1}+O(n^{-3}).
\end{eqnarray*}

 \item 
\begin{equation}
 \frac{d\lambda_\pm}{d\theta}
=
2\sin\tau \cdot
\left(
\pm \frac{1}{2}\lambda_{0,1} h^{-1/2}+ \lambda_{0,2} 
\pm \frac{3}{2}\lambda_{0,3} h^{1/2}+ 2\lambda_{0,4} h
\pm \frac{5}{2}\lambda_{0,5} h^{3/2} + O(h^2)
\right).
\label{pu14}
\end{equation}

 \item 
\begin{eqnarray}
 \frac{d^2\lambda_\pm}{d\theta^2}
=\pm \lambda_{2,-3} h^{-3/2}\pm \lambda_{2,-1} h^{-1/2}
+ \lambda_{2,0} \pm \lambda_{2,1} h^{1/2}
+ O(h),
\label{pu19}
\end{eqnarray}
\begin{eqnarray}
  \lambda_{2,-3}
&=&\frac{V^2}{2}(n\pi)^{-1} + O(n^{-3}),
\label{pu21}\\
 \lambda_{2,-1}
&=&
-\frac{V^2}{16}(n\pi)^{-1}+O(n^{-3}),
\label{pu22}\\
 \lambda_{2,0}
&=&
2 +\frac{V^2}{6}(n\pi)^{-2}+ O(n^{-4})\label{pu23},\\
 \lambda_{2,1}
&=&
-\frac{9V^2}{256}(n\pi)^{-1}
+O(n^{-3}).
\label{pu24}
\end{eqnarray}

 \item  
\begin{equation}
  \frac{d^3\lambda_\pm}{d\theta^3}
=2\sin \tau\cdot
\left(
\mp\frac{3}{2} \lambda_{2,-3} h^{-5/2}
\mp\frac{1}{2} \lambda_{2,-1} h^{-3/2}
\pm\frac{1}{2} \lambda_{2,1} h^{-1/2}
+ O(1)
\right).\label{pu25}
\end{equation}
\end{enumerate}
\end{proposition}
\begin{proof}
 The proof can be done simply by substituting the expansion 
(\ref{ba07}), (\ref{asymptotic_ln}) 
and (\ref{e1})-(\ref{e5}) into $\lambda_{\pm}=k_\pm^2$
and taking the derivative with respect to 
$\theta$ (or $\tau$) repeatedly.
The formulas (\ref{formula1})-(\ref{formula5}) help the calculation.
\end{proof}

By construction, we have
\begin{equation}
 \lambda_+(\theta)=\lambda_{n+1}(\theta),\quad
 \lambda_-(\theta)=\lambda_{n}(\theta),\quad
\theta=n\pi+\tau
\label{ba09}
\end{equation}
for small $\tau$.
Using (\ref{ba09}) and the expansion
(\ref{pu19}), 
we can find the asymptotics 
of the solution to $\lambda_n''(\theta)=0$.
\begin{proposition}
 \label{proposition_inflection}
The equation
\begin{displaymath}
\frac{d^2\lambda_n}{d\theta^2}=0,\quad
(n-1)\pi \leq \theta \leq n\pi
\end{displaymath}
has a unique solution $\theta=\theta_0$ for sufficiently large $n$.
The asymptotics of $\theta_0$ is given as
\begin{equation}
\label{asymptotics_theta0}
 \theta_0 = n\pi - \left(\frac{V^2}{4n\pi}\right)^{1/3}
+
O(n^{-1})\quad
\mbox{as}\ n\to \infty.
\end{equation}
\end{proposition}
Before the proof,
we give a numerical result
by using the explicit formulas (\ref{der1})-(\ref{der6}),
in Figure \ref{figure_inflection}.
The result shows
the formula (\ref{asymptotics_theta0}) gives a good approximation.
\begin{figure}[htbp]
\begin{center}
 \includegraphics[width=6cm]{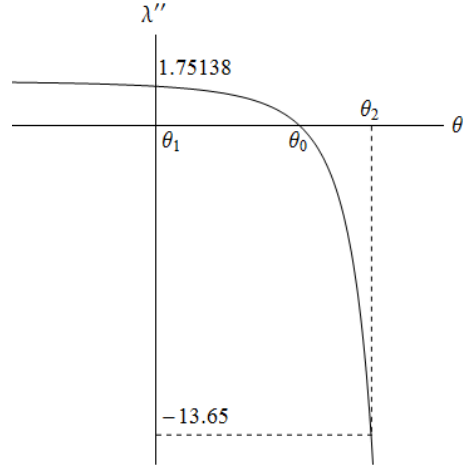}
\caption{The graph of $\lambda_n''(\theta)$
near $\theta_0=n\pi-\delta_n^{1/3}$,
$\delta_n=V^2/(4n\pi)$.
Here we take $V=1$, $n=1000$, 
$\theta_1=n\pi - 2 \delta_n^{1/3}$,
and $\theta_2=n\pi - \delta_n^{1/3}/2$.
}
\label{figure_inflection}
\end{center} 
\end{figure}

\begin{proof}
Put $\delta_n=V^2/(4n\pi)$.
We divide the interval $[(n-1)\pi,n\pi]$
into three intervals
\begin{eqnarray*}
&& I_1=[(n-1)\pi,(n-1)\pi+\delta_n^{1/4}],\quad
 I_2=[(n-1)\pi+\delta_n^{1/4}, n\pi -\delta_n^{1/4}],\\
&& I_3=[n\pi -\delta_n^{1/4}, n\pi].
\end{eqnarray*}
For $\theta\in I_1$, we apply the first equality of (\ref{ba09}) 
with $n$ replaced by $n-1$, so $\lambda_n=\lambda_+$
and $\theta=(n-1)\pi + \tau$. 
Then the expansion (\ref{pu19}) for $d^2\lambda_+/d\theta^2$
implies $\lambda_n''(\theta)>0$ for $\theta\in I_1$,
for sufficiently large $n$.
Moreover,
the formula (\ref{inner_lambda2}) given later 
in Proposition \ref{proposition_inner_taylor}
shows $\lambda_n''(\theta)>0$ for $\theta\in I_2$,
for sufficiently large $n$.

For $\theta\in I_3$, we apply 
the second equality of (\ref{ba09}), so $\lambda_n=\lambda_-$
and $\theta=n\pi+\tau$.
Then $-\delta_n^{1/4}<\tau<0$,
and (\ref{pu25}) for $\lambda_-$  implies
$\lambda_n^{(3)}(\theta)<0$ 
for $\theta\in I_3\setminus\{n\pi\}$,
for sufficiently large $n$.
Thus $\lambda_n''(\theta)$ is monotone decreasing on this interval,
and (\ref{pu19}) for $d^2\lambda_-/d\theta^2$ implies
$\lambda_n''(n\pi -\delta_n^{1/4})>0$ and 
$\lambda_n''(n\pi )<0$ for sufficiently large $n$.%
\footnote{When $V>0$, the latter fact can also be proved
by the explicit value
$\lambda_n''(n\pi)=-4(n\pi)^2/V$;
see ((\ref{pm11}) below).
}
Thus 
there exists a unique $\theta_0\in [(n-1)\pi,n\pi]$
with $\lambda_n''(\theta_0)=0$
for sufficiently large $n$,
and $\theta_0\in I_3$.

Let us find more detailed asymptotics of $\theta_0$.
By (\ref{pu19}) for $d^2\lambda_-/d\theta^2$,
\begin{eqnarray}
 \frac{d^2\lambda_n}{d\theta^2}=0
&\Leftrightarrow&
- \lambda_{2,-3} h^{-3/2}- \lambda_{2,-1} h^{-1/2}
+ \lambda_{2,0} - \lambda_{2,1} h^{1/2}
+ O(h)=0\nonumber\\
&\Leftrightarrow&
h
= f(h),\quad
f(h)=
\left(
\frac{\lambda_{2,-3} + \lambda_{2,-1} h
+  \lambda_{2,1} h^2
+ O(h^{5/2})}
{\lambda_{2,0}}
\right)^{2/3}.
\label{ba11}
 \end{eqnarray}
If $\theta\in I_3$, then $\tau =O(n^{-1/4})$
and $h=O(n^{-1/2})$ by (\ref{formula2}).
So, if $h$ is the solution to (\ref{ba11}),
the expansions of the coefficients (\ref{pu21})-(\ref{pu24})
imply $h=O(n^{-2/3})$,
and again by (\ref{ba11})
\begin{eqnarray}
h = f(h)&=&
 \left(\delta_n+O(n^{-5/3})\right)^{2/3}\nonumber\\
&=& \delta_n^{2/3}
\left(1+O(n^{-2/3})\right).
\label{ba12}
\end{eqnarray}
Thus (\ref{ba12}) and (\ref{formula2}) imply
\begin{displaymath}
 \tau = -\delta_n^{1/3}+O(n^{-1}).
\end{displaymath}
Since $\theta=n\pi + \tau$, we have the conclusion.
\end{proof}

\subsection{Analysis on $\lambda_n(\theta)$ near the band center}
Let us prove 
$\lambda_n(\theta)$ is similar to the parabola $\lambda=\theta^2$
near the band center, that is, for $\theta$ in the interval
\begin{equation}
 \label{inner06}
\left[
(n-1)\pi+\left(\frac{V^2}{4n\pi}\right)^{1/4}
,
n\pi-\left(\frac{V^2}{4n\pi}\right)^{1/4}
\right].
\end{equation}
Notice that the asymptotics (\ref{asymptotics_theta0})
implies $\theta_0$ is not in the interval (\ref{inner06})
for large $n$.

\begin{lemma}
 \label{lemma_taylor_inner}
Let $k(\theta)=D^{-1}(2\cos\theta)$ given in (\ref{ba01}),
and $\theta$ in the interval (\ref{inner06}).
Then,
we have the following expansion.
\begin{eqnarray}
\label{taylor_k_inner} 
k &=& \theta - V D'(\theta)^{-1}\frac{\sin\theta}{\theta}
-\frac{V^2}{2}D'(\theta)^{-3}D''(\theta)
\left(\frac{\sin\theta}{\theta}\right)^2
+ R_3,\\
R_3 &=&
-V^3\left(\frac{\sin\theta}{\theta}\right)^3
\int_0^1 \frac{(1-z)^2}{2}\left(D^{-1}\right)^{(3)}
\left(D(\theta) -z V\frac{\sin\theta}{\theta}\right)dz.
\nonumber
\end{eqnarray}
\end{lemma}
{\bf Remark.} The explicit forms of the derivatives
are given in (\ref{d0k})-(\ref{d6k}).

\begin{proof}
The Taylor expansion of $k=D^{-1}(y)$
around $y=D(\theta)$ is 
\begin{eqnarray}
 k 
&=& D^{-1}(D(\theta))
+ (D^{-1})'(D(\theta))\cdot(y-D(\theta))\nonumber\\
&&\quad + \frac{(D^{-1})''(D(\theta))}{2}\cdot(y-D(\theta))^2
+ R_3(y),\label{ba15}\\
R_3(y)&=&
(y-D(\theta))^3\cdot
\int_0^1 \frac{(1-z)^2}{2}
\left(D^{-1}\right)^{(3)}\left(D(\theta)+ z(y-D(\theta))\right)dz.
\label{ba16}
\end{eqnarray}
Substituting the equality
\begin{displaymath}
y=2\cos\theta = D(\theta)-V\frac{\sin\theta}{\theta}
\end{displaymath}
into (\ref{ba15}) and (\ref{ba16}),
we obtain the conclusion since
\begin{eqnarray*}
D^{-1}(D(\theta))&=&\theta,\\
 (D^{-1})'(D(\theta))&=&D'(\theta)^{-1},\\
 (D^{-1})''(D(\theta))&=&-D'(\theta)^{-3}D''(\theta),
\end{eqnarray*}
where we used (\ref{dkformula2}).
\end{proof}
\begin{proposition}
\label{proposition_inner_taylor}
 Let $n$ be a sufficiently large integer.
Then,
for $\theta$ in the interval (\ref{inner06}),
we have the expansion of $k=D^{-1}(2\cos\theta)$
given in (\ref{ba01}) as follows.
\begin{eqnarray}
\label{inner_k0}
k
&=&
\theta
+ \frac{V}{2\theta}
+ \frac{V^2}{8\theta^2} \cot\theta
+ O(n^{-5/2})
,\\
\nonumber
\frac{dk}{d\theta}
&=& 
1 - \frac{V}{2\theta^2} - \frac{V^2}{8\theta^2\sin^2\theta}+
O(n^{-9/4})
,\\
\nonumber
\frac{d^2k}{d\theta^2}
&=&   
\frac{V^2\cos\theta}{4\theta^2\sin^3\theta}+
O(n^{-2}).
\end{eqnarray}
Moreover, for $\lambda_n(\theta)=(k(\theta))^2$ we have
\begin{eqnarray}
\nonumber
\lambda_n
&=& \theta^2 + V 
+\frac{V^2}{4\theta}\cot\theta
+ O(n^{-3/2}),\\
\nonumber
\frac{d\lambda_n}{d\theta}
&=& 
2\theta 
-\frac{V^2}{4\theta \sin^2\theta}
+O(n^{-5/4}),\\
  \label{inner_lambda2}
\frac{d^2\lambda_n}{d\theta^2}
&=& 
2
+
\frac{V^2\cos\theta}{2\theta\sin^3\theta}
+
O(n^{-1}).
\end{eqnarray}
\end{proposition}
\begin{proof}
Let us rewrite the formula (\ref{taylor_k_inner}) as
\begin{eqnarray}
k = \theta 
- V \frac{\sin\theta}{D'(\theta)}\frac{1}{\theta}
-\frac{V^2}{2}
\left(
\frac{\sin \theta}{D'(\theta)}
\right)^3
\frac{D''(\theta)}{\theta^2\sin\theta}
+ R_3.
\label{inner15}
\end{eqnarray}
When $\theta$ is in the interval (\ref{inner06}), we have
\begin{equation}
 \label{inner16}
\theta = O(n),\quad \theta^{-1}=O(n^{-1}),\quad
(\sin\theta)^{-1}=O(n^{1/4}),\quad
\cot\theta = O(n^{1/4}).
\end{equation}
Thus
\begin{eqnarray}
 \frac{\sin\theta}{D'(\theta)}
&=&
\frac{\sin\theta}{-2\sin\theta + V(\theta^{-1}\cos\theta-\theta^{-2}\sin\theta)}
\nonumber\\
&=&-\frac{1}{2}\cdot
\frac{1}{1 - (V/2)(\theta^{-1}\cot\theta-\theta^{-2})}\nonumber\\
&=&-\frac{1}{2}-\frac{V}{4}\theta^{-1}\cot\theta + O(n^{-3/2}),
\label{inner17}
\\
 \left(\frac{\sin\theta}{D'(\theta)}\right)^3
&=&-\frac{1}{8}-  \frac{3V}{16}
\theta^{-1}\cot\theta + O(n^{-3/2}),
\nonumber\\
 \frac{D''(\theta)}{\sin\theta}
&=&
\frac{-2\cos\theta +V(-\theta^{-1}\sin\theta 
- 2 \theta^{-2}\cos\theta + 2\theta^{-3}\sin\theta)}{\sin\theta}\nonumber\\
&=& -2\cot\theta + O(n^{-1}).\nonumber
\end{eqnarray}
Thus the first three terms in (\ref{inner15}) coincide
the formula (\ref{inner_k0}).
Let us show that the remainder term $R_3$ is negligible.
By the differentiation of the inverse function,
we can prove that
\begin{equation}
\label{inner20}
(D^{-1})^{(\j)}(y)=D'(k)^{-(2j-1)}\cdot
(\mbox{Polynomial of }D'(k),\ \ldots,\ D^{(j)}(k)).
\end{equation}
The polynomial part is bounded uniformly with respect to $n$.
By the expansion (\ref{puiseux_k}) and the monotonicity of $k=k(\theta)$,
we see that $k=k(\theta)$
satisfies
\begin{equation}
\label{inner21}
 (n-1)\pi + \left(\frac{V^2}{4n\pi}\right)^{1/4}+O(n^{-3/4})
\leq k
\leq n\pi - \left(\frac{V^2}{4n\pi}\right)^{1/4}+O(n^{-3/4}),
\end{equation}
and
\begin{equation}
\label{inner22}
 D'(k)^{-1}=(-2\sin k + O(n^{-1}))^{-1} = O(n^{1/4}).
\end{equation}
Put $k_z=D^{-1}(D(\theta)-zV{\sin\theta}/{\theta})$ $(0\leq z \leq 1)$,
then $k_z$ lies between $k=D^{-1}(2\cos\theta)$ and $\theta=D^{-1}(D(\theta))$,
and by (\ref{inner20})-(\ref{inner22})
\begin{eqnarray*}
&&(D^{-1})^{(3)}\left(D(\theta)-zV\frac{\sin\theta}{\theta}\right)\nonumber\\
&=& 
D'(k_z)^{-5}\cdot 
(\mbox{Polynomial of }D'(k_z),D''(k_z), D^{(3)}(k_z))\nonumber\\
&=& O(n^{5/4}).
\end{eqnarray*}
Thus, we once have a rough estimate
\begin{displaymath}
 |R_3|= O(n^{-3})\cdot O(n^{5/4}) = O(n^{-7/4}),
\end{displaymath}
so the equation (\ref{inner_k0}) holds with the worse remainder term
$O(n^{-7/4})$. However, this conclusion implies
\begin{equation}
 \label{inner24}
 |k(\theta)-\theta|=O(n^{-1}),
\end{equation}
which also implies
\begin{eqnarray*}
&& (\sin \theta)^3 \cdot 
(D^{-1})^{(3)}\left(D(\theta)-zV\frac{\sin\theta}{\theta}\right)\\
&=& 
\left(\frac{\sin\theta}{D'(k_z)}\right)^3
\cdot D'(k_z)^{-2}\cdot 
(\mbox{Polynomial of }D'(k_z),D''(k_z), D^{(3)}(k_z))\\
&=& O(n^{1/2}),
\end{eqnarray*}
Thus we have $|R_3|=O(n^{-5/2})$ and (\ref{inner_k0}) holds.

Other estimates can be obtained by differentiating the 
formula (\ref{taylor_k_inner}).
Then we find the estimate for the remainder term becomes worse 
by the power $n^{1/4}$ per one differentiation.
For example, the leading term in the remainder 
in (\ref{inner17}) is up to constant multiple
\begin{displaymath}
 \theta^{-2}\cot^2\theta = O(n^{-3/2}).
\end{displaymath}
Differentiating this term, we get
\begin{displaymath}
-2 \theta^{-3}\cot^2\theta+\theta^{-2}\cdot 2\cot\theta
\cdot \left(-\frac{1}{\sin^2\theta}\right)=O(n^{-5/4}),
\end{displaymath}
and the result is worse than $O(n^{-3/2})$ by $n^{1/4}$.
As for $R_3$,
\begin{eqnarray}
 \frac{dR_3}{d\theta}
&=&
-V^3\left(\frac{3\sin^2\theta\cos\theta}{\theta^3}-\frac{3\sin^3\theta}{\theta^4}\right)
\int_0^1 \frac{(1-z)^2}{2}\left(D^{-1}\right)^{(3)}
\left(D(\theta) -z V\frac{\sin\theta}{\theta}\right)dz\nonumber\\
&&-V^3\left(\frac{\sin\theta}{\theta}\right)^3
\int_0^1 \frac{(1-z)^2}{2}\left(D^{-1}\right)^{(4)}
\left(D(\theta) -z V\frac{\sin\theta}{\theta}\right)\nonumber\\
&&\hspace{3cm}\cdot
 \left(D'(\theta)-zV\left(\frac{\cos\theta}{\theta}-\frac{\sin\theta}{\theta^2}\right)\right)
dz.\label{inner25}
\end{eqnarray}
For the first term of (\ref{inner25}), 
one $\sin \theta$ in the numerator changed into $\cos\theta$ 
by differentiation,
and the estimate becomes worse by $n^{1/4}$,
because of (\ref{inner16}).
For the second term,
that `$(D^{-1})^{(3)}$ turned into $(D^{-1})^{(4)}$' 
makes two $D'(k_z)$ in the denominator (see (\ref{inner20})),
one of which cancels with $D'(\theta)$ appeared next.
Thus the estimate also becomes worse by $n^{1/4}$, in total.
We can treat the other remainder terms similarly.
\end{proof}

\subsection{Estimate for the amplitude function}
We shall give the bound for the amplitude
function in (\ref{bloch16}), that is,
\begin{eqnarray}
a_n(\theta,x,y)
&=&\Phi_n(\theta,x,y)\frac{dk}{d\theta}\nonumber\\
&=&\frac{-2\sin k}{D'(k)}
\left(
\cos kx + \frac{V}{2k}\sin kx
\right)
\left(
\cos ky + \frac{V}{2k}\sin ky
\right)\nonumber\\
&&\hspace{.5cm}+ i 
\frac{-2\sin \theta}{D'(k)}
\sin k(x-y)
+
\frac{-2\sin^2 \theta}{\sin k \cdot D'(k)}\sin kx \sin ky,
\label{amp01} 
\end{eqnarray}
where $x,y\in (0,1)$ and $k=k(\theta)=D^{-1}(2\cos \theta)$
given in (\ref{ba01}).

\begin{proposition}
 \label{proposition_amplitude1}
Let $a_n(\theta,x,y)$ given in (\ref{amp01}).
\begin{enumerate}
 \item The function $a_n$ is bounded uniformly with respect to
$\theta\in \mathbf{R}$, $x,y\in (0,1)$ and $n=1,2,\ldots$.

 \item Put $\delta_1 = |V|/(n\pi)$, $\delta_2=(V^2/(4n\pi))^{1/4}$, and
\begin{eqnarray*}
&&I_1 = [(n-1)\pi, (n-1)\pi +\delta_1],\quad
I_2 = [(n-1)\pi + \delta_1, (n-1)\pi + \delta_2],\\
&& 
I_3 = [(n-1)\pi + \delta_2, n\pi - \delta_2],\quad
I_4 = [n\pi - \delta_2, n\pi - \delta_1],\quad
I_5 = [n\pi - \delta_1, n\pi].
\end{eqnarray*}
For sufficiently large $n$,
the derivative $a_n'={\partial a_n}/{\partial\theta}$ 
obeys the following bound
\begin{equation}
|a_n'(\theta,x,y)|
\leq
\begin{cases}
 C n & (\theta\in I_1),\\
 C (n^{-1}(\theta-(n-1)\pi)^{-2}+1) & (\theta\in I_2),\\
 C (n^{-1/2}+1)& (\theta\in I_3),\\
 C (n^{-1}(n\pi - \theta)^{-2} +1)& (\theta\in I_4),\\
 C n & (\theta\in I_5),
\end{cases}
 \label{bound_amplitude}
\end{equation}
where 
$C$ is a positive constant independent of 
$\theta$, $x$, $y$, and $n$.
Especially,
\begin{displaymath}
\|a_n'\|_{L^1([(n-1)\pi, n\pi])}
= \int_{(n-1)\pi}^{n\pi}|a_n'(\theta,x,y)|d\theta
\end{displaymath}
is bounded uniformly with respect to $n$, $x$, $y$.
\end{enumerate}
\end{proposition}
\begin{proof}
(i) First we prove $dk/d\theta=-2\sin \theta/D'(k)$ is 
bounded uniformly with respect to $n$ and 
$\theta\in [(n-1)\pi,n\pi)]$.
For $\theta\in I_4\cup I_5$, we have $\tau = \theta-n\pi = O(n^{-1/4})$,
and $h = (-1)^n (y_n-2\cos\theta)=O(n^{-1/2})$ by (\ref{formula2}).
Then we have by the expansions (\ref{taylor_dk}) and (\ref{ba07})
\begin{eqnarray}
k-l_n&=& k_{-} -l_n = - e_1 h^{1/2}(1+O(n^{-1/2})) = 
- h^{1/2}(1+O(n^{-1/2})),\nonumber\\
 D'(k) &=& 2d_2 (k-l_n)(1+O(n^{-1/2}))= (-1)^n h^{1/2}(1+O(n^{-1/2})).
\label{amp02}
\end{eqnarray}
Let $z_n=(-1)^n y_n$ as in (\ref{zndef}). 
Since $z_n \geq 2$, 
\begin{equation}
\label{amp021}
 h
=z_n-2\cos\tau
\geq 
2 - 2 \cos \tau
= 4\sin^2\frac{\tau}{2}.
\end{equation}
Thus
\begin{displaymath}
\left| \frac{-2\sin\theta}{D'(k)}\right|
\leq 
\left|
\frac{2\sin\tau}{2 \sin(\tau/2)}
\right|
(1+O(n^{-1/2}))=
\left|
2\cos(\tau/2)
\right|
(1+O(n^{-1/2})),
\end{displaymath}
and the right hand side is uniformly bounded.
We can prove $dk/d\theta$ is uniformly bounded
for $\theta\in I_1\cup I_2$ in a similar way.
For $\theta\in I_3$,
the bounds (\ref{inner16}) and (\ref{inner24}) hold,
so (\ref{inner16}) holds even if $\theta$ is replaced by $k$.
Then
\begin{eqnarray}
 \frac{-2\sin\theta}{D'(k)}
&=&\frac{-2\sin k (1+O(n^{-1}))}{-2\sin k + V(k^{-1}\cos k - k^{-2}\sin k)}
\nonumber\\
&=&\frac{1+O(n^{-1})}{1 -(V/2)(k^{-1}\cot k - k^{-2})}
= 1 + O(n^{-3/4}).
\label{amp03}
\end{eqnarray}
Thus $dk/d\theta$ is uniformly bounded on all the intervals $I_1,\ldots,I_5$.
Similarly we can prove 
$-2 \sin k/D'(k) $ and $-2 \sin^2\theta/(\sin k \cdot D'(k))$ 
are uniformly bounded.
The remaining factors are clearly bounded,
so we have the conclusion.

(ii) It is sufficient to show the three functions
\begin{displaymath}
f_1(\theta)= \frac{d}{d\theta}\left(\frac{-2\sin k}{D'(k)}\right),\quad
f_2(\theta)= \frac{d}{d\theta}\left(\frac{-2\sin \theta}{D'(k)}\right),\quad
f_3(\theta)= \frac{d}{d\theta}\left(\frac{-2\sin^2 \theta}{\sin k \cdot D'(k)}\right)
\end{displaymath}
obey the bound (\ref{bound_amplitude}),
since the derivatives of the remaining factors are bounded.

First, by (\ref{d1k}), (\ref{d2k}) and (\ref{bloch17.5})
\begin{eqnarray}
f_1(\theta)
&=&
\frac{-2\cos k D'(k)+ 2 \sin k D''(k)}{D'(k)^2}\cdot \frac{dk}{d\theta}
\nonumber\\
&=& \frac{-2\sin \theta}{D'(k)^3}\cdot
V(-2k^{-1}-2k^{-2}\cos k \sin k + 4 k^{-3}\sin^2 k)\nonumber\\
&=& \frac{-2\sin \theta}{D'(k)^3}\cdot O(n^{-1}).
\label{amp04}
\end{eqnarray}
For $\theta \in I_4\cup I_5$,
the expansion (\ref{amp02}) implies
\begin{eqnarray}
 \frac{-2\sin \theta}{D'(k)^3} = -2 \cdot \sin\tau \cdot h^{-3/2} (1+O(n^{-1/2})).
\label{amp05}
\end{eqnarray}
For $\theta \in I_5$, by (\ref{formula1})
\begin{eqnarray}
&&\sin \tau = O(n^{-1}),\nonumber\\
&& h = z_n - 2 \cos\tau \geq z_n -2
= \frac{V^2}{4}(n\pi)^{-2}+O(n^{-4}),\quad
h^{-3/2} = O(n^3).
\label{amp06}
\end{eqnarray}
Thus (\ref{amp04}), (\ref{amp05}) and (\ref{amp06}) imply
 $|f_1(\theta)|=O(n)$ for $\theta\in I_5$.
For $\theta \in I_4$, we have 
$|\tau|=O(n^{-1/4})$ and
by (\ref{amp021})
\begin{eqnarray}
&& \sin \tau = \tau (1+O(n^{-1/2})),\nonumber\\
&& h \geq  \tau^2 (1+O(n^{-1/2})),\quad
h^{-3/2} \leq  \tau^{-3}(1+O(n^{-1/2}))
\label{amp07}
\end{eqnarray}
for large $n$.
Thus (\ref{amp04}), (\ref{amp05}) and (\ref{amp07}) imply
$|f_1(\theta)|\leq C n^{-1}\tau^{-2}$ for $\theta\in I_4$,
for some positive constant $C$ independent of $n$.
For $\theta\in I_3$, we have by (\ref{amp03})
\begin{displaymath}
 -\frac{2\sin \theta}{D'(k)} = O(1),\quad D'(k)^{-1}=O(n^{1/4}).
\end{displaymath}
So (\ref{amp04}) implies $|f_1(\theta)|=O(n^{-1/2})$ for $\theta\in I_3$.
Similarly we can prove $|f_1(\theta)|= O(n^{-1}(\theta-(n-1)\pi)^{-2})$
for $\theta\in I_2$ and $|f_1(\theta)|= O(n)$ for $\theta\in I_1$,
thus $f_1(\theta)$ obeys the bound (\ref{bound_amplitude}).

Next, we shall estimate $f_2(\theta)$.
By (\ref{dkformula1}),
(\ref{d0k})-(\ref{d2k})
and (\ref{bloch17.5}), we have
\begin{eqnarray*}
 f_2(\theta)
&=&\frac{-2\cos\theta\cdot  D'(k)
+2\sin\theta \cdot D''(k)\cdot(-2\sin\theta)/D'(k)}{D'(k)^2}\\
&=&\frac{-D(k)D'(k)^2-(4-D(k)^2)D''(k)}{D'(k)^3}\\
&=& \frac{1}{D'(k)^3}
(-2V^2k^{-2}\cos k + O(n^{-3}))\\
&=&\frac{1}{D'(k)^3}\cdot O(n^{-2}).
\end{eqnarray*}
This equality and (\ref{amp02}), (\ref{amp06}), (\ref{amp07}) and
(\ref{inner16}) (with $\theta$ replaced by $k$) 
gives the same conclusion for $f_2(\theta)$
(actually, we obtain a bit faster decay for $\theta\in I_2\cup I_3 \cup I_4$).

Finally, by (\ref{dkformula1}), (\ref{d0k})-(\ref{d2k}) and (\ref{bloch17.5}),
\begin{eqnarray*}
&& \frac{-2\sin^2\theta}{\sin k \cdot D'(k)}
=
\frac{D(k)^2-4}{2\sin k \cdot D'(k)}
= \frac{-4 \sin k + 4Vk^{-1}\cos k + V^2 k^{-2}\sin k}{2D'(k)},\\
&&f_3(\theta)
=
\frac{-2\sin \theta}{D'(k)^3}(2Vk^{-1}+ O(n^{-2}))
=\frac{-2\sin \theta}{D'(k)^3}\cdot O(n^{-1}).
\end{eqnarray*}
This estimate is the same as (\ref{amp04}).
So the same conclusion holds for $f_3(\theta)$.
\end{proof}

\section{Proof of Theorem \ref{theorem_main}}
The $L^1-L^\infty$ norm of the operator
$P_ne^{-itH}$ is just the supremum 
with respect to $x,y\in \mathbf{R}\setminus\mathbf{Z}$
of the absolute value of the 
integral kernel $K_{n,t}(x,y)$ given in (\ref{bloch16}).
Put $x=[x]+x'$, $y=[y]+y'$, $[x],[y]\in \mathbf{Z}$,
$x',y'\in (0,1)$,
and $s=([x]-[y])/t$, then (\ref{bloch16}) is rewritten as%
\footnote{This kind of modification is used in the analysis by
Korotyaev \cite{Ko}, 
in which the propagation speed of the wave front 
in a periodic media is studied.
}
\begin{eqnarray}
 K_{n,t}(x,y)
=
\frac{1}{2\pi}
\int_{-\pi}^\pi e^{-it (\lambda_n(\theta)-s\theta)}
a_n(\theta,x',y')
d\theta,
\label{st01}
\end{eqnarray}
where $a_n$ is the amplitude function given in (\ref{amp01}).
The following lemma, taken from Stein's book \cite[page 334]{St},
gives the decay order of the oscillatory integral with respect to $t$.
 
\begin{lemma}
 \label{lemma_stein}
Let $I=(a,b)$ be a finite open interval and $k=1,2,3,\ldots$.
Let $\phi\in C^{k}(I;\mathbf{R})$, and assume
\begin{displaymath}
m_k=\inf_{x\in I}|\phi^{(k)}(x)|>0.
\end{displaymath}
If $k=1$, we additionally assume $\phi'$ is a 
monotone function on $I$.
Let $\psi\in C^1(I;\mathbf{C})$ and assume $\psi'\in L^1(I)$.
Then we have
\begin{displaymath}
\left|\int_a^b e^{i t \phi(x)}\psi(x)dx\right|\leq 
t ^{-1/k}
\cdot 
c_k (m_k)^{-1/k}
\left(|\psi(b)| + \int_a^b |\psi'(x)|dx\right)
\end{displaymath}
for any $t>0$, where $c_k = 5\cdot 2^{k-1}-2$.
\end{lemma}
{\bf Remark.}
Here, we say $f(x)$ is monotone if $f(x)\leq f(y)$ ($x<y$)
or $f(x)\geq f(y)$ ($x<y$).
The assumption `$I$ is finite' can be removed if the integral 
in the left hand side converges.
The proof is done by integration by parts
and a mathematical induction (see Stein \cite[page 332-334]{St}).

\begin{proof}[Proof of Theorem \ref{theorem_main}]
The statement for $|t|\leq 1$ immediately follows from
Proposition \ref{proposition_amplitude1}.
By taking the complex conjugate if necessary,
we can assume $t>1$ without loss of generality.

Let $\theta_0$ as in Proposition \ref{proposition_inflection}.
Let $v_{\max} = |\lambda_n'(\theta_0)|$ be the maximum of
the function $|\lambda_n'(\theta)|$.
Put $\theta_0=n\pi + \tau_0$ and 
$h_0 = (-1)^n (y_n -2\cos \theta_0) = z_n- 2\sin \tau_0$.
Put $\delta_n=V^2/(4n\pi)$.
Then we have by (\ref{asymptotics_theta0}) and (\ref{formula2}),
\begin{eqnarray}
\label{pm01}
&& \tau_0 =  - \delta_n^{1/3} + O(n^{-1}),\\
&& \sin \tau_ 0 = - \delta_n^{1/3} + O(n^{-1}),
\label{pm02}
\\
&& h_0 = \delta_n^{2/3} + O(n^{-4/3}).
\label{pm03}
\end{eqnarray}
Since $\lambda_n(\theta)=\lambda_-(\theta)$ near the upper edge,
we have by (\ref{pm01})-(\ref{pm03}) and (\ref{pu14})%
\footnote{
The formula (\ref{pm04})
is consistent with Figure \ref{figure_lambda1}.
}
\begin{eqnarray}
 v_{\max}&=&\lambda_n'(\theta_0)\nonumber\\
&=& 2\sin \tau_0\cdot
\left(
- \frac{1}{2}\lambda_{0,1} h_0^{-1/2}+ \lambda_{0,2} 
-\frac{3}{2}\lambda_{0,3} h_0^{1/2}+ 2\lambda_{0,4} h_0
- \frac{5}{2}\lambda_{0,5} h_0^{3/2} + O(h_0^2)
\right)\nonumber\\
&=& 2 n \pi \left(1 + O(n^{-2/3})\right).
\label{pm04}
\end{eqnarray}

If $|s|\geq v_{\max}+1$, then we have
\begin{equation}
\label{pm05}
 |\lambda_n'(\theta)-s|\geq |s| - v_{\max} \geq 1.
\end{equation}
By the periodicity of the integrand,
we can divide the integral (\ref{st01})
as the sum of two integrals so that
$\lambda_n'$ is a monotone function on each interval.
Then we can apply Lemma \ref{lemma_stein} with $k=1$,
and we have by (\ref{pm05}) and Proposition \ref{proposition_amplitude1}
\begin{eqnarray*}
 \left|K_{n,t}(x,y)\right|&\leq& 2 t^{-1 }c_1 
\cdot 
\left(
\inf_{\theta\in \mathbf{R}}|\lambda_n'(\theta)-s|
\right)^{-1}
\left(
\|a\|_{L^\infty(I)}
+
\|a'\|_{L^1([I)}
\right)\\
&\leq& C t^{-1},
\end{eqnarray*}
where $I=[-\pi,\pi]$ and $C$ is a positive constant independent of $n$.

If $|s|\leq v_{\max}+1$,
then we slide the interval of integration by periodicity,
and divide it into four intervals
\begin{eqnarray*}
I_1 = \left[
n\pi- 2 \delta_n^{1/3}, n\pi - \delta_n^{1/3}/2
\right],\quad
I_2 =  \left[
 n\pi - \delta_n^{1/3}/2,  n\pi + \delta_n^{1/3}/2
\right],\quad\\
I_3 =  \left[n\pi + \delta_n^{1/3}/2, n\pi + 2 \delta_n^{1/3}
\right],\quad
I_4 =\left[n\pi + 2 \delta_n^{1/3}, (n+2)\pi -2 \delta_n^{1/3}\right].
\end{eqnarray*}
Put
\begin{displaymath}
 K_j = \frac{1}{2\pi}\int_{I_j}e^{-it(\lambda_n(\theta)-s\theta)}
a_n(\theta,x',y')d\theta\quad
(j=1,2,3,4).
\end{displaymath}
If $\theta\in I_1$ or $I_3$, 
then $\delta_n^{1/3}/2\leq |\tau|\leq 2 \delta_n^{1/3}$
and by (\ref{formula2})
\begin{eqnarray}
&&h = \tau^2 + O(\tau^4)+O(n^{-2}) =\tau^2 (1+O(n^{-2/3})),\nonumber\\
&&h^{-1} = \tau^{-2}(1+O(n^{-2/3})).
\label{pm06}
\end{eqnarray}
By (\ref{pu19}),
\begin{eqnarray}
\lambda_n''(\theta)
&=& -2 \delta_n h^{-3/2}+ 2 + O(n^{-2/3})\nonumber\\
&=& -2 \delta_n |\theta-n\pi|^{-3}+ 2+O(n^{-2/3})
\quad (\theta\in I_1\cup I_3).
\label{pm07}
\end{eqnarray}
(\ref{pm06}) and (\ref{pm07}) imply%
\footnote{
We can make sure the accuracy of the formulas
(\ref{pm08}) and (\ref{pm09}) 
in Figure \ref{figure_inflection}.
}
\begin{eqnarray}
 \lambda_n''(n\pi \pm 2 \delta_n^{1/3})=\frac{7}{4} + O(n^{-2/3}),
\label{pm08}
\\
 \lambda_n''(n\pi \pm \delta_n^{1/3}/2)=-14 + O(n^{-2/3}).
\label{pm09}
\end{eqnarray}
Next, for $0<|\tau|<\delta_n^{1/4}$, 
we have from (\ref{formula2})
\begin{displaymath}
  \frac{V^2}{4}(n\pi)^{-2} + O(n^{-4}) \leq
h \leq \tau^2+O(n^{-1}),
\end{displaymath}
and by (\ref{pu25}) 
\begin{eqnarray}
 \lambda_n^{(3)}(\theta)
= 
\sin \tau \cdot
\frac{3V^2}{2n\pi} \cdot h^{-5/2}
\left(1+O(h)+ n h^{5/2}O(1)\right)
\not = 0\quad
(0<|\tau|<\delta_n^{1/4})
\label{pm10}
\end{eqnarray}
for sufficiently large $n$,
since $h=O(n^{-1/2})$ and $n h^{5/2}\cdot O(1)=O(n^{-1/4})$.
Thus $\lambda_n^{(2)}$ is monotone on the left half of $I_2$.
This fact and (\ref{pm09}) imply
\begin{equation}
 |\lambda_n''(\theta)|\geq |\lambda_n''(n\pi \pm \delta_n^{1/3}/2)|
\geq C\quad (\theta\in I_2)
\end{equation}
for sufficiently large $n$,
for some positive constant $C$ independent of $n$.
Moreover, (\ref{pm08}), (\ref{pm10}) 
and (\ref{inner_lambda2}) imply 
$|\lambda_n''(\theta)|$ is also uniformly bounded from below
on $I_4$.
Then Lemma \ref{lemma_stein} implies $K_2$
and $K_4$ is $O(t^{-1/2})$, uniformly with respect to $n$.
Moreover, for $\theta\in I_1\cup I_3$,
we have $\delta_n^{1/3}/2\leq |\tau|\leq 2\delta_n^{1/3}$,
and (\ref{formula2}) and (\ref{pm10}) imply
\begin{equation}
 |\lambda_n^{(3)}(\theta)|
= 6\delta_n \tau^{-4} (1+O(\tau^2))
\geq \frac{3}{8}\delta_n^{-1/3} (1+O(n^{-2/3}))
\geq C n^{1/3}
\end{equation}
for sufficiently large $n$,
where $C$ is a constant independent of $n$.
Thus Lemma \ref{lemma_stein} implies
\begin{displaymath}
 |K_1|+|K_3| \leq C t^{-1/3}n^{-1/9},
\end{displaymath}
and the conclusion holds.
\end{proof}

We conclude the paper by arguing
the summability of the bandwise estimates.
If we fix $x$, $y$ and take the limit $n\to \infty$, 
then $|s| \leq v_{\max}$ for sufficiently large $n$,
by (\ref{pm04}) (see also Figure \ref{figure_lambda1}).
So there always exists the stationary phase point $\theta_s$
(the solution to $\lambda_n'(\theta)=s$)
nearby $\theta=n \pi$,
for sufficiently large $n$.
For simplicity, assume $V>0$ in the sequel.
When $\theta=n\pi$, we have by the formulas
(\ref{der1})-(\ref{der6}),
(\ref{dkformula2}),
(\ref{d1k}) and (\ref{d2k})
\begin{eqnarray}
&&k(n\pi)=n\pi,\quad k'(n\pi)=k^{(3)}(n\pi)=0,\nonumber\\
&& k''(n\pi)=D'(n\pi)^{-1}(-2 \cos n\pi)
= n\pi (V \cos n\pi)^{-1}\cdot (-2 \cos n\pi)= -\frac{2n\pi}{V},
\nonumber\\
&&\lambda_n'(n\pi)=\lambda_n^{(3)}(n\pi)=0,\nonumber\\
&&\lambda_n''(n\pi) = 2(k(n\pi)k''(n\pi)+ k'(n\pi)^2)=-\frac{4(n\pi)^2}{V}.
\label{pm11}
\end{eqnarray}
Thus, 
even if we cut out a small interval $J_n$ around $\theta=n\pi$,
the bound for the integral over $J_n$ is not better than 
$O(n^{-1}t^{-1/2})$,
which is not summable with respect to $n$.
However,
we already know the sum $\sum_{n=1}^\infty P_n e^{-itH}$
converges in the strong operator topology in $L^2(\mathbf{R})$.
Thus it seems that the sum of the integral kernels 
converges only conditionally.
We have to analyze
the cancellation between the integral kernels for two adjacent bands
more carefully,
in order to obtain a better estimate.
We hope to argue this subject in the future.

\end{document}